\documentclass[11pt]{article}

\newcommand{\COMMENTED}[1]{{}}
\newcommand{\junk}[1]{\COMMENTED{#1}}
\newcommand{\hide}[1]{{}}

\title{\textbf{Graph Sparsification for Derandomizing Massively Parallel Computation with Low Space}}

\author{\textbf{Artur Czumaj}
	\thanks{Department of Computer Science and Centre for Discrete Mathematics and its Applications (DIMAP), University of Warwick. Email: A.Czumaj@warwick.ac.uk. Research partially supported by the Centre for Discrete Mathematics and its Applications (DIMAP), by a Weizmann-UK Making Connections Grant, by IBM Faculty Award, and by EPSRC award EP/N011163/1.}
	\and
	\textbf{Peter Davies}
	\thanks{Institute of Science and Technology Austria (IST Austria). Email: peter.davies@ist.ac.at. Work partially done at the Department of Computer Science and Centre for Discrete Mathematics and its Applications (DIMAP), University of Warwick.  Research partially supported by the European Union's Horizon 2020 research and innovation programme under the Marie Sk\l odowska-Curie grant agreement No 754411, the Centre for Discrete Mathematics and its Applications, a Weizmann-UK Making Connections Grant, and EPSRC award EP/N011163/1.}
	\and
	\textbf{Merav Parter}
	\thanks{Weizmann Institute of Science, Rehovot, Israel. Email: merav.parter@weizmann.ac.il. Research partially supported by a Weizmann-UK Making Connections Grant.}
}
\usepackage[T1]{fontenc}
\usepackage[utf8]{inputenc}
\usepackage{lmodern}
\usepackage[a4paper]{geometry}

\usepackage{amsmath,amsthm,amssymb}
\newtheorem{theorem}{Theorem}
\newtheorem{corollary}[theorem]{Corollary}
\newtheorem{lemma}[theorem]{Lemma}

\newtheorem{definition}[theorem]{Definition}

\usepackage{algorithm}
\usepackage{algpseudocode}

\usepackage{bm}
\usepackage{xcolor}
\usepackage{fullpage}
\usepackage{xspace}
\usepackage{paralist}

\usepackage{ifpdf}
\newcommand{\mydriver}{hypertex}
\ifpdf
\renewcommand{\mydriver}{pdftex}
\fi
\usepackage[breaklinks,\mydriver]{hyperref}
\usepackage{cleveref}
\newcommand{\congest}{\textsf{CONGEST}\xspace}
\newcommand{\congc}{\textsf{CONGESTED CLIQUE}\xspace}
\newcommand{\local}{\textsf{LOCAL}\xspace}
\newcommand{\MPC}{\textsf{MPC}\xspace}

\newcommand{\cj}{\ensuremath{c}\xspace}
\newcommand{\Exp}[1]{\mathbf{E}\left[#1\right]}
\newcommand{\Prob}[1]{\mathbf{Pr}\left[#1\right]}	
\newcommand{\nat}{\ensuremath{\mathbb{N}}}

\def\epsilon{\ensuremath{\varepsilon}}
\newcommand{\eps}{\ensuremath{\epsilon}}
\newcommand{\machines}{\ensuremath{M}}
\newcommand{\spac}{\ensuremath{S}}

\newcommand{\CE}{\ensuremath{\mathcal{E^*}}}
\newcommand{\E}{\ensuremath{\mathcal{E}}}
\newcommand{\A}{\ensuremath{\mathcal{A}}}
\newcommand{\B}{\ensuremath{\mathcal{B}}}
\newcommand{\J}{\ensuremath{\mathcal{Q}}}
\newcommand{\IS}{\ensuremath{\mathcal{I}}}
\newcommand{\MIS}{\ensuremath{\textsf{MIS}}}

\junk{
	\author{\textbf{Artur Czumaj}
		\thanks{Department of Computer Science and Centre for Discrete Mathematics and its Applications (DIMAP), University of Warwick. Email: A.Czumaj@warwick.ac.uk. Research partially supported by the Centre for Discrete Mathematics and its Applications (DIMAP), by a Weizmann-UK Making Connections Grant, by IBM Faculty Award, and by EPSRC award EP/N011163/1.}
		\and
		\textbf{Peter Davies}
		\thanks{Institute of Science and Technology Austria (IST Austria). Email: peter.davies@ist.ac.at. Work partially done at the Department of Computer Science and Centre for Discrete Mathematics and its Applications (DIMAP), University of Warwick.  Research partially supported by the European Union's Horizon 2020 research and innovation programme under the Marie Sk\l odowska-Curie grant agreement No 754411, the Centre for Discrete Mathematics and its Applications, a Weizmann-UK Making Connections Grant, and EPSRC award EP/N011163/1.}
		\and
		\textbf{Merav Parter}
		\thanks{Weizmann Institute of Science, Rehovot, Israel. Email: merav.parter@weizmann.ac.il. Research partially supported by a Weizmann-UK Making Connections Grant.}
}}

\begin{document}
%---------------------------------------------------------------------------------------------------------------------------------------------------------

\begin{titlepage}

\maketitle
\begin{abstract}
The Massively Parallel Computation (\MPC) model is an emerging model which distills core  aspects of distributed and parallel computation. It has been developed as a tool to solve (typically graph) problems in systems where the input is distributed over many machines with limited space.

Recent work has focused on the regime in which machines have sublinear (in $n$, the number of nodes in the input graph) memory, with randomized algorithms presented for fundamental graph problems of Maximal Matching and Maximal Independent Set. However, there have been no prior corresponding \emph{deterministic} algorithms.

A major challenge underlying the sublinear space setting is that the local space of each machine might be too small to store all the edges incident to a single node. This poses a considerable obstacle compared to the classical models in which each node is assumed to know and have easy access to its incident edges. To overcome this barrier we introduce a new \emph{graph sparsification technique} that \emph{deterministically} computes a low-degree subgraph with additional desired properties. The degree of the nodes in this subgraph is small in the sense that the edges of each node can be now stored on a single machine. This low-degree subgraph also has the property that solving the problem on this subgraph provides \emph{significant} global progress, i.e., progress towards solving the problem for the original input graph.

Using this framework to derandomize the well-known randomized algorithm of Luby [SICOMP'86], we obtain $O(\log \Delta+\log\log n)$-round \emph{deterministic} MPC algorithms for solving the fundamental problems of \emph{Maximal Matching} and \emph{Maximal Independent Set} with $O(n^{\epsilon})$ space on each machine for any constant $\epsilon > 0$. Based on the recent work of Ghaffari et al. [FOCS'18], this additive $O(\log\log n)$ factor is \emph{conditionally} essential. These algorithms can also be shown to run in $O(\log \Delta)$ rounds in the closely related model of \congc, improving upon the state-of-the-art bound of $O(\log^2 \Delta)$ rounds by Censor-Hillel et al. [DISC'17].
\end{abstract}

%--------------------------------------------------------------------------------------------------------------------------------------------------------- 
\end{titlepage}

\section{Introduction}
\label{sec:Intro}

%---------------------------------------------------------------------------------------------------------------------------------------------------------

%---------------------------------------------------------------------------------------------------------------------------------------------------------
The last few years have seen an increasing interest in the design of parallel algorithms. This has been largely caused by the successes of a number of massively parallel computation frameworks, such as MapReduce~\cite{DG04,DG08}, Hadoop~\cite{White12}, Dryad~\cite{IBYBF07}, or Spark~\cite{ZCFSS10}, which resulted in the need of active research for understanding the computational power of such systems. The \emph{Massively Parallel Computations (\MPC)} model, first introduced by Karloff et al.\ \cite{KSV10} have became the standard theoretical model of algorithmic study, as it provides a clean abstraction of these frameworks. % (cf. \cite{ANOY14,BKS13,BKS14,GSZ11,KSV10}).

The \MPC model shares many similarities to earlier models of parallel computation, for example with the PRAM model; indeed, it was quickly observed that it is easy to simulate a single step of PRAM in a constant number of rounds on \MPC~\cite{GSZ11,KSV10}, implying that a vast body of work on PRAM algorithms naturally translates to the \MPC model.
However, the fact that the \MPC model allows for a lot of local computation (in principle, unbounded) enabled it to capture a more ``coarse-grained'' and meaningful aspect of parallelism. Recent works have brought a number of new algorithms for fundamental graph combinatorial and optimization problems that demonstrated that in many situations the \MPC model can be significantly more powerful than PRAM. And so, for example, in a sequence of papers we have seen that the fundamental problems of connectivity (see, e.g., \cite{ASSWZ18,ASW19}), matching, maximal independent set, vertex cover, coloring, etc (see, e.g., \cite{ABBMS17,ASZ20,BBDFHKU19,BDELM19,BDELMS19,BHH19,chang2019complexity,CLMMOS18,GGKMR18,GU19,LMOS20,LMSV11}) can be solved in the \MPC model significantly faster than on the PRAM. However, the common feature of most of these results is that they were relying on randomized algorithms, and very little research has been done to study deterministic algorithms.

The main theme of this paper is to \emph{explore the power of the \MPC model in the context of deterministic algorithms}. In particular, we want to understand whether the \MPC model allows faster deterministic algorithms than in the PRAM-like models, in a similar way as it has demonstrated to do in the setting of randomized computations.

We consider two corner-stone problems of local computation:  \emph{maximal matching} and \emph{maximal independent set (MIS)}. These problems are arguably among the most fundamental graph problems in parallel and distributed computing with numerous %theoretical and also practical implications.
applications. The study of these problems can be traced back to PRAM algorithms of the 1980s \cite{ABI86,II86,KW85,Luby86} and they have been considered as benchmark problems in various computational models since then. In particular, these problems have been central in our understanding of derandomization techniques. Luby \cite{Luby86}, and independently Alon et al.\ \cite{ABI86}, have been the first to present a generic transformation of parallel algorithms for maximal matching and MIS, to obtain efficient deterministic algorithms for these problems in the PRAM model. For example, Luby \cite{Luby86} showed that his randomized MIS $O(\log n)$-time algorithm can be derandomized on PRAM in $O(\log^3 n \log\log n)$ time. The bound was later improved to $O(\log^3 n)$ time \cite{GS89}, $O(\log^{2.5} n)$ time \cite{Han96}, and then $\widetilde{O}(\log^2 n)$ time \cite{Har19}.

\paragraph{The \MPC model.}

The \emph{Massively Parallel Computations (\MPC)} model was first introduced by  Karloff et al.\ \cite{KSV10} and later refined in \cite{ANOY14,BKS13,GSZ11}. In the \MPC model, there are $\machines$ machines and each of them has $\spac$ words of space. Initially, each machine receives its share of the input. In our case, the input is a collection $V$ of nodes and $E$ of edges and each machine receives approximately $\frac{n+m}{\machines}$ of them (divided arbitrarily), where $|V| = n$ and $|E| = m$.

The computation proceeds in synchronous \emph{rounds} in which each machine processes its local data and performs an arbitrary local computation on its data without communicating with other machines. At the end of each round, machines exchange messages. Each message is sent only to a single machine specified by the machine that is sending the message. All messages sent and received by each machine in each round have to fit into the machine's local memory. Hence, their total length is bounded by $\spac$. This, in particular, implies that the total communication of the \MPC model is bounded by $\machines \cdot \spac$ in each round. The messages are processed by recipients in the next round.
At the end of the computation, machines collectively output the solution. The data output by each machine has to fit in its local memory. Hence again, each machine can output at most $\spac$ words.

%---------------------------------------------------------------------------------------------------------------------------------------------------------

\paragraph{Local memory, range of values for $\spac$ and $\machines$, and fully scalable algorithms.}
If the input is of size $N$, one usually wants $\spac$ to be sublinear in $N$, and the total space across all the machines to be at least $N$ (in order for the input to fit onto the machines) and ideally not much larger. Formally, one usually considers $\spac \in \Theta\left(N^{\eps}\right)$, for some $\eps > 0$. Optimally, one would want to design \emph{fully scalable algorithms}, which work for \emph{any} positive value of $\eps$ (see, e.g., \cite{ASSWZ18,GU19,Onak18}), though most of the earlier works have been focusing on graph algorithms whose space $\spac$ has been close to the number of edges of the graph (see, e.g., \cite{LMSV11}, where $\spac = \Theta\left(n^{1+\eps}\right)$), or close to the number of nodes of the graph (see, e.g., \cite{ABBMS17,CLMMOS18,GGKMR18}, where $\spac = \widetilde{\Theta}(n)$).

In this paper, the focus is on graph algorithms. If $n$ is the number of nodes in the graph, the input size can be as large as $\Theta\left(n^2\right)$. Our deterministic parallel algorithms are fully scalable, i.e., for any constant $\eps > 0$ require $\spac = \Theta\left(n^{\eps}\right)$ space per machine, which is polynomially less than the size of the input. % for either sparse or dense graphs.

\paragraph{Known bounds.}
For many graph problems, including MIS and maximal matching, fully scalable \emph{randomized} $O(\log n)$ round $n^{\Omega(1)}$ space \MPC algorithms can be achieved by simulating PRAM algorithms \cite{ABI86,II86,Luby86}. These bounds (still in the randomized case) have been improved only very recently and only in some settings. For fully scalable algorithms, we know only of a \emph{randomized} algorithm due to Ghaffari and Uitto \cite{GU19} working in $\widetilde{O}(\sqrt{\log\Delta})$ rounds for maximal matching and MIS, where $\Delta$ is the maximum degree. Better bounds are known for maximal matching algorithms using significantly more memory: Lattanzi et al.\ \cite{LMSV11} gave an $O(1/\eps)$ rounds randomized algorithm using $O(n^{1+\eps})$ space per machine, and Behnezhad et al.\ \cite{BHH19} presented an $O(\log\log\Delta + \log\log\log n)$ rounds randomized algorithm in $n/2^{\Omega(\sqrt{\log n})}$ space.
Further, Ghaffari et al.\ \cite{GKU19} gave conditional evidence that no $o(\log\log n)$ round fully scalable \MPC algorithm can find a maximal matching, or MIS (or a constant approximation of minimum vertex cover or of maximum matching).

Unfortunately, much less is known about \emph{deterministic} \MPC algorithms for maximal matching and MIS. Except some parts of the early work in \cite{GSZ11} (cf. Lemma \ref{lem:comm}), we are not aware of any previous algorithms designed specifically for the \MPC model. One can use a simulation of PRAM algorithms (so long as they use $\tilde O(m)$ total space) to obtain fully scalable deterministic algorithms for maximal matching and MIS on \MPC, and their number of rounds would be asymptotically the same; to our knowledge, the fastest deterministic PRAM algorithms require $O(\log^{2.5}n)$ \cite{Han96} rounds for maximal matching, and $\widetilde{O}(\log^2 n)$ rounds for MIS \cite{Har19}. If one can afford to use linear memory per machine, $\spac = O(n)$, then the recent deterministic \congc algorithms for MIS by Censor-Hillel et al.\ \cite{CPS17}, directly give an $O(\log n \log \Delta)$-round deterministic \MPC algorithm for MIS \junk{(see also Appendix~\ref{subsubsec:comp-Censor-Hillel-et-al})}.

There have been some related works on derandomization in the \local, \congest, and \congc distributed models (cf. \cite{BKM19,DKM19,GHK18,GK18,Harris18,Parter18,PY18}). The \MPC setting in the low space regime brings along crucial challenges that distinguishes it from the all previous computational models in which derandomization has been studied. The inability of a machine to view the entire neighborhood a node requires novel derandomization paradigms, which is why we developed deterministic graph sparsification. We note that graph sparsification has been shown to be useful before in the context of low-space \emph{randomized} \MPC algorithms (e.g.,~\cite{GU19}).

%It should be also mentioned that there have been some related works on derandomization in the \congc model \cite{BKM19,Parter18,PY18} studying deterministic spanners and vertex coloring. % (e.g., very recently, Bamberger et al.\ \cite{BKM19} gave a deterministic \congc algorithm that solves the $(\Delta+1)$-coloring (and the $(\text{degree} + 1)$-list-coloring problem) in $O(\log^3\Delta)$ rounds).
%Similarly, there has been some recent work on derandomization in the \local and \congest models, see, e.g., \cite{BKM19,DKM19,GHK18,GK18,Harris18}. The \MPC setting in the low space regime brings along crucial challenges that distinguishing it from the all previous computational models in which derandomization has been studied. The inability of a machine to view the entire neighbor list of a node calls for developing a new derandomization paradigm that is based on deterministic graph sparsification. We note that graph sparsification has been shown to be useful before in the context low-space \MPC algorithms (e.g., \cite{GU19}), but only in the randomized setting.

%---------------------------------------------------------------------------------------------------------------------------------------------------------

\subsection{New results}

We demonstrate the power of the deterministic algorithms in the \MPC model on the example of two fundamental optimization problems: finding a maximal matching and finding an MIS.

%In this paper we present fully scalable deterministic algorithms that find a maximal matching and MIS in $O(\log\Delta + \log\log n)$ rounds in the \MPC model. The following are our two main results.

\begin{theorem}[Maximal Matching and MIS]
\label{thm:MM-MIS-Delta}
For any constant $\eps>0$, \textbf{maximal matching} and \textbf{MIS} can be found deterministically in the \MPC model in $O(\log\Delta + \log\log n)$ rounds, using $O(n^{\eps})$ space per machine and $O(m+n^{1+\eps})$ total space.
\end{theorem}

While the additive $O(\log\log n)$ term in the bound in Theorem \ref{thm:MM-MIS-Delta} looks undesirable, it is most likely necessary. Indeed, as mentioned earlier, Ghaffari et al.\ \cite{GKU19} provided an $\Omega(\log\log n)$ conditional hardness result for maximal matching and MIS, even for \emph{randomized} fully scalable \MPC algorithms. They proved that unless there is an $o(\log n)$-round (randomized) \MPC algorithm for connectivity with local memory $\spac = n^{\eps}$ for a constant $0 < \eps < 1$ and $\text{poly}(n)$ global memory (see \cite{RVW18} for strong arguments about the hardness of that problem), there is no component-stable randomized \MPC algorithm with local memory $\spac = n^{\eps}$ and $\text{poly}(n)$ global memory that computes a maximal matching or an MIS in $o(\log \log n)$ rounds (see \cite{GKU19} for a more precise claim and explanations of the notation). For maximal matching, the lower bound holds even on trees.

%---------------------------------------------------------------------------------------------------------------------------------------------------------

\subsubsection{Our approach: deterministic graph sparsification}
\label{subsubsec:approach}

%We outline the main concepts behind our work, and how they are applied to the problems we consider:

%\paragraph{Deterministic graph sparsification.}

Our results rely on a generic method of computing deterministically low-degree subgraphs with some desired properties that depend on the specific problem. This approach has two key objectives:
\begin{inparaenum}[(i)]
\item providing a randomized sparsification procedure that uses only pairwise (or limited) independence, and thus based on a small seed (e.g., of $O(\log \Delta)$ bits),
\item providing an efficient derandomization technique using the classical methods of conditional expectations. % (cf. \cite{Luby93}). %applied in \cite{CPS17,GK18,PY18}.
\end{inparaenum}
Compared with previous work our challenge is twofold:
\begin{inparaenum}[(i)]
\item handling the situation where nodes do not know all their neighbors and
\item beating the existing bounds known for \MPC (and also \congc) models
by using a more tailored derandomization approach (e.g., reducing seed length).
\end{inparaenum}

The most obvious limitation in low-space \MPC is that a node's neighborhood cannot be collected onto a single machine. To circumvent this, we derandomize the sampling of a low-degree subgraph, so that neighborhoods (in fact, 2-hop neighborhoods) in the subgraph \emph{do} fit in the memory of a single machine. The challenge is to show that we can obtain low-degree subgraphs which preserve properties useful for the problem. Specifically, we need to find subgraphs that contain matchings and independent sets which are adjacent to a constant fraction of edges in the original graph.

It transpires that the constraint of having space per machine $\spac$ means that, in a single \emph{stage}, we can only derandomize the sub-sampling of nodes or edges with probability at least (roughly) $\spac^{-1}$. Otherwise, any fixed machine holds fewer than one sampled node or edge in expectation, so machines cannot tell which random seeds are good when performing the method of conditional expectations. If we wish to subsample with a lower probability, we must split the process into multiple stages using probabilities at least $\spac^{-1}$. Since we are assuming that $\spac = n^{\Theta(1)}$, though, only a constant number of such stages is required. So, we are able to ``deterministically sample,'' as long as we are able to maintain some invariant that guarantees that our sampled graph has some good property after every stage.

\junk{
Another caveat is that we can only subsample while node degrees are at least $n^\delta$, for some constant $\delta>0$. Beyond this point, concentration bounds for $c$-wise independent random variables (for constant $c$) are insufficient to achieve high probability, which means that we cannot maintain properties for all nodes in the sampled graph. We could overcome this by using higher independence, but then our families of hash functions would need to be larger, and we would no longer be able to apply the method of conditional expectations in a constant number of rounds. Fortunately, for our purposes, sampling down to degree $n^\delta$ is sufficient, since we can then fit neighborhoods (and 2-hop neighborhoods) %into the space $\spac$
on a single machine.
}

Another caveat is that we can only subsample node degrees until they are at least $n^\delta$, for some constant $\delta > 0$. Beyond this point, concentration bounds for $c$-wise independent random variables (for constant $c$) will not guarantee high probability bounds, and thus we cannot maintain properties for all nodes in the sampled graph. If we used higher independence then our families of hash functions would need to be larger and the method of conditional expectations would require a super-constant number of rounds. Fortunately, sampling down to degree $n^\delta$ is sufficient in out setting, since we can then fit 2-hop neighborhoods on a single machine.

This \emph{deterministic sampling} is the cornerstone of both of our algorithms. Having obtained, deterministically, a low degree graph which maintains certain good properties, we can then perform one more derandomization step to obtain a maximal matching or MIS. This step is similar to Luby's algorithm, and essentially involves sampling edges (or nodes) with probability inversely proportional to their degree in the sampled graph. We need to have 2-hop neighborhoods stored on a single machine in order for it to determine which edges/nodes are in the matching or independent set. As mentioned earlier, using $c$-wise independence we cannot achieve high-probability bounds for this process, but in this case constant probability is sufficient. This is because we no longer need to find a seed that is good for every node; we merely need to ensure that the seed induces a matching or independent set which will remove a constant fraction of the edges from the original graph.

This overall process can be performed in only a constant number of \MPC rounds in total, and constructs a matching/independent set such that removing the set and its neighborhood from the graph reduces the number of edges by a constant factor. After $O(\log n)$ iterations (i.e., $O(\log n)$ \MPC rounds overall), no edges remain in the graph, so we have found an MIS or maximal matching, respectively.

To improve the round complexity to $O(\log \Delta+\log\log n)$, we note that if $\log\Delta$ is $o(\log n)$ then neighborhoods of radius
%$O(\log n/\log\Delta)$ %$O(\log_{\Delta}n)$ 
$O\left(\frac{\log n}{\log \Delta}\right)$
already fit onto single machines, and we do not have to perform our deterministic graph sparsification step. Furthermore, we show that a stage of Luby's algorithm requires only an $O(\log \Delta)$-bit random seed, and so we can derandomize
%$O(\log n/\log\Delta)$ $O(\log_{\Delta}n)$ 
$O\left(\frac{\log n}{\log \Delta}\right)$
stages of Luby's algorithm at once using the method of conditional expectations. This allows us to \emph{compress} the stages of the algorithm, and perform all $\Theta(\log n)$ necessary stages in only $O(\log \Delta)$ \MPC rounds. The $O(\log\log n)$ term arises from the need to collect the neighborhoods onto machines.

\paragraph{Maximal independent set and matching.}

Our end-goal result converts a version of Luby's \emph{randomized} MIS and maximal matching  algorithms into $O(\log n)$-round fully scalable \emph{deterministic} algorithms in the \MPC model. Luby's algorithm takes an input graph and in $O(\log n)$ rounds finds an MIS. In each round, the algorithm finds some independent set $\IS$ and removes $\IS$ and $N(\IS)$ from the input graph. The crux of the analysis is that in expectation, the number of edges removed in a single round is a constant fraction of the original number of edges, ensuring that in expectation,
$O(\log n)$ rounds suffice to find an MIS. A similar construction holds for the maximal matching problem, where one finds and removes matchings rather than independent sets.

It is known that Luby's algorithm can be derandomized, however on PRAMs and other classical models of parallel and distributed computation models, PRAMs and \congc, the known derandomization techniques are expensive and lead to deterministic algorithms running in a \emph{super-logarithmic number of rounds} (see, e.g., \cite{ABI86,CPS17,GS89,Han96,Har19,Luby86}).

In this paper, we will show that the \MPC model can perform the derandomization very efficiently, even in the fully scalable model with very limited memory per machine, via our generic method of derandomization of the sampling of a low-degree graph while maintaining some desired properties. We will model Luby's algorithm and ensure that each of its rounds can be simulated deterministically in a constant number of \MPC rounds. In particular, we present a deterministic $O(1)$-\MPC round construction that for a given graph $G$ finds an independent set $\IS$ for which the number of edges incident to $\IS \cup N(\IS)$ is at least constant fraction of the edges of $G$. (A similar construction holds for maximal matching.)

We implement a variant of Luby's algorithm whose randomness can be limited by using only $c$-wise independent hash functions, for some constant $c$. Then, we employ the classic method of conditional expectations to find a hash function that will ensure sufficient progress (in our case, that we find an independent set whose removal eliminates a constant fraction of edges from the graph). While this approach may appear standard, it requires overcoming several major challenges caused not only by the deterministic use of hash functions, but most importantly by the limitations of the fully scalable model with very limited memory per machine in the \MPC model.
%
%\paragraph{Maximal Independent Set and Matching:}

%\paragraph{Difficulties to overcome in low-memory \MPC}
%
%%
%\paragraph{From $O(\log n)$ to $O(\log \Delta)$.}

%---------------------------------------------------------------------------------------------------------------------------------------------------------

\subsubsection{Implications to \congc}
\label{implic-congc}

As recently observed (cf. \cite{BDH18}), the \MPC model is very closely related to the \congc model from distributed computing. The nowadays classical distributed \congc model (see, e.g., \cite{LPP15,Peleg00}) is a variant of the classical \local model of distributed computation (where in each round network nodes can send through all incident links messages of unrestricted size) with limited communication bandwidth. The distributed system is represented as a complete network (undirected graph) $G$, where network nodes execute distributed algorithms in synchronous rounds. In any single round, all nodes can perform an unlimited amount of local computation, send a possibly different $O(\log n)$-bit message to each other node, and receive all messages sent by them. We measure the complexity of an algorithms by the number of synchronous rounds required.

It is not difficult to see (see, e.g., \cite{BDH18}) that any $r$-round \congc algorithm can be simulated in $O(r)$ rounds in the \MPC model with $n$ machines and $\spac = O(rn)$. Furthermore, Behnezhad et al.\ \cite{BDH18} showed that by using the routing scheme of Lenzen~\cite{Lenzen13}, \MPC algorithms with $\spac = O(n)$ are adaptable to the \congc model. These results immediately imply that the recent deterministic \congc algorithm due to Censor-Hillel et al.\ \cite{CPS17} to find MIS in $O(\log n \log \Delta)$ rounds can be extended to be run in the \MPC model with $\spac = \tilde O(n)$. (When $\Delta = O(n^{1/3})$, the bound improves to $O(\log \Delta)$.) Notice though, that in contrast to our work, the derandomization algorithm from \cite{CPS17} relies on a derandomization of Ghaffari's MIS algorithm \cite{Ghaffari16}, whereas our derandomization is based on Luby's MIS algorithm. % (see Section \ref{subsubsec:comp-Censor-Hillel-et-al} for more details).

These simulations imply also that our new deterministic \MPC algorithms for maximal matching and MIS can be implemented to run in the \congc model using $O(\log \Delta)$ rounds. By combining Theorem \ref{alg:MIS}, for the regime $\Delta = \omega(n^{1/3})$, with the $O(\log \Delta)$-round MIS algorithm of \cite{CPS17} for the regime $\Delta = O(n^{1/3})$, we get the $O(\log \Delta)$-round algorithm for MIS. We further note that, in the $\Delta = O(n^{1/3})$ regime, one can collect 2-hop neighborhoods onto single machines, and thus find a maximal matching by simulating MIS on the line graph of the input graph. So, combining Theorem \ref{thm:MM-MIS-Delta} with the MIS algorithm of \cite{CPS17} yields the following:

\begin{corollary}
\label{cor:congcMISMM}
One can deterministically find MIS and maximal matching in $O(\log\Delta)$ rounds in the \congc model.
\end{corollary}

%---------------------------------------------------------------------------------------------------------------------------------------------------------

\def\APPENDCOMPARISONWITHCENSORHILLELETAL{
%\subsubsection{Comparison with Censor-Hillel et al.\ \cite{CPS17}}
%\label{subsubsec:comp-Censor-Hillel-et-al}
\paragraph{Comparison with Censor-Hillel et al.\ \cite{CPS17}}

%\Artur{This will be either completely removed of heavily edited; there is no space.}%
Our framework of derandomization bares some similarity to a recent approach due to Censor-Hillel et al.\ \cite{CPS17} which deterministically finds MIS in $O(\log n \log \Delta)$ rounds in the \congc model; by known reductions this extends to the \MPC model with $\spac = \tilde{O}(n)$. Notice that the new algorithm in our paper obtains a smaller number of rounds $O(\log \Delta)$ and smaller space on each machine $\spac = n^{\eps}$.
%. As mentioned earlier, by known reductions this approach can be extended to be run in the \MPC model with $\spac = \tilde{O}(n)$. (Notice that the new algorithm presented in this paper obtains a smaller number of rounds $O(\log \Delta)$ and significantly smaller space on each machine $\spac = n^{\eps}$.)

The deterministic MIS algorithm of Censor-Hillel et al.\ \cite{CPS17} works in the \congc model and it is based on the derandomization of Ghaffari's MIS algorithm \cite{Ghaffari16}. The latter algorithm of \cite{Ghaffari16} has the common two-phase structure: a randomized part of $O(\log \Delta)$ rounds followed by a deterministic part that solves the remaining undecided graph in $2^{O(\sqrt{\log n})}$ rounds. \cite{CPS17} first show that a slight modification of Ghaffari's randomized part can be simulated using only pairwise independence. As a result, each of the $O(\log \Delta)$ randomized steps can be simulated using a random seed of length $O(\log n)$. By combining with the method of conditional expectations of Luby \cite{Luby93}, \cite{CPS17} gave an  $O(\log n)$ round procedure to compute the $O(\log n)$ seed for each of the $O(\log \Delta)$ steps, yielding a total round complexity of $O(\log \Delta \log n)$ rounds. On the high level, the seed is computed in a bit-by-bit manner spending $O(1)$ rounds for each bit using a \emph{voting} procedure (i.e., each node votes for its preferable bit value). \cite{CPS17} also showed that if each node knows its $2$-hop neighborhood, the round complexity can be considerably improved to $O(1)$ rounds per randomized step, yielding in this case an $O(\log \Delta)$-round algorithm for the entire MIS computation.
%
%
%Essentially to compute the $O(\log n)$ seed for a single step in Ghaffari's algorithm, the nodes use a voting based approach. Specifically, each node sends its vote for a single bit in the seed to a global leader and the elected value of the bit is the one that got most votes. In this way, the bits of the seed are computed in a bit-by-bit manner, spending $O(1)$ rounds for the computation of each bit. The actual voting procedure of Censor-Hillel et al.\ \cite{CPS17} is more involved as it also has to make sure that after $O(\log \Delta)$ steps of derandomization the remaining unsolved graph has $O(n)$ edges \footnote{To prove the shattering phenomenon in Ghaffari's algorithm pairwise independence seems insufficient, and thus Censor-Hillel et al.\ \cite{CPS17} had to come up with a more involved voting procedure.}
%

The approach taken in the current paper bares some similarity with that of \cite{CPS17}, but it is also different in several key aspects. The main limitation in the derandomization of Ghaffari's algorithm in the \MPC setting is the following. In \cite{CPS17} it is shown that after $O(\log \Delta)$ steps of derandomization the remaining unsolved graph has a linear number of edges. In the \congc model, at this point, the computation can be completed in just a constant number of rounds by collecting the unsolved graph to a single node using Lenzen's routing algorithm \cite{Lenzen13}. In our \MPC setting, the memory of each machine is sublinear and thus this approach is no longer applicable.

Our algorithm is based on derandomizing Luby's MIS algorithm, spending only $O(1)$ rounds to simulate each randomized step rather than $O(\log n)$ rounds as follows by the approach of \cite{CPS17}. This fast derandomization is in particular challenging in the case of large degrees, specifically for $\Delta = n^{\Omega(\epsilon)}$. In the latter case, one cannot even collect the neighbors of a node into a single machine. This challenge calls for a different approach than that taken in \cite{CPS17}, based on graph sparsification. Roughly speaking, in each Luby's step, our algorithm puts a focus on a subset of promising nodes, namely, nodes whose removal (by either joining the MIS or having their neighbors join the MIS) reduces the size of the remaining unsolved graph by a constant factor. It then provides a deterministic procedure to sparsify this subset in a way that guarantees that the induced subgraph on the selected nodes is sparse. This sparsity allows us to collect the $2$-hop neighborhood of each such node into its machine. At this point, we can apply a faster $O(1)$-round derandomization in a similar manner to that used\footnote{In \cite{CPS17}, this faster derandomization was applied only for $\Delta = O(\sqrt{n})$, namely, in a setting where a node can learn its $2$-hop neighbors in a constant number of rounds.} in \cite{CPS17}.  We note that the approach of identifying a sparse subgraph whose removal yields to a large (global) progress has been used before in the setting of \emph{randomized} local algorithms (see e.g., \cite{Ghaffari17,GU19}) but it is considerably less studied in the context of deterministic algorithms.

\junk{The algorithm of Censor-Hillel et al.\ \cite{CPS17} is based on the derandomization of Ghaffari's MIS algorithm \cite{Ghaffari16}. In \cite{CPS17}, it is shown that after $O(\log \Delta)$ steps of derandomization the remaining unsolved graph has $O(n)$ edges. If $\spac = \tilde{O}(n)$, the computation can be completed in a constant number of rounds by collecting the unsolved graph to a single node using Lenzen's routing algorithm~\cite{Lenzen13}. In our \MPC setting with $\spac = O(n^{\eps})$, this approach is no longer applicable.

Our algorithm is based on derandomizing Luby's MIS algorithm, spending only $O(1)$ rounds to simulate each randomized step rather than $O(\log n)$ rounds as follows by the approach of \cite{CPS17}. This fast derandomization is in particular challenging in the case of large degrees, specifically for $\Delta = n^{\Omega(\epsilon)}$. Since one cannot even collect the neighbors of a node into a single machine, we incorporate a different approach than that taken in \cite{CPS17}, based on graph sparsification.
}
}
\APPENDCOMPARISONWITHCENSORHILLELETAL
%---------------------------------------------------------------------------------------------------------------------------------------------------------

\section{Preliminaries}
\label{sec:pre}

%We introduce some notation we will use throughout, and then outline our approach and the existing tools we will require for it.

%\subsection{Notation}
\label{sec:notation}

An \emph{independent set} in a graph $G = (V,E)$ is any subset of nodes $\IS \subseteq V$ such that no two nodes in $\IS$ share an edge. An independent set $\IS$ is called a \emph{maximal independent set (MIS)} if it is not possible to add any other node of $G$ to $\IS$ and obtain an independent set. %; and it is a \emph{maximum independent set} if every independent set in $G$ has size at most $|\IS|$.

A \emph{matching} of a graph $G = (V,E)$ is any independent subset of edges $M \subseteq E$ (i.e., no two edges in $M$ share an endpoint). A matching $M$ of a graph $G$ is a \emph{maximal matching} if it is not possible to add any other edge of $G$ to $M$ and obtain a matching. %; and it is a \emph{maximum matching} if every matching in $G$ has size at most $|M|$.

For a node $v \in V$, the neighborhood $N(v)$ is the set of nodes $u$ with $\{u,v\} \in E$; for any $U \subseteq V$, we define $N(U) = \bigcup_{v \in U} N(v)$.

In any graph $G$ we denote the degree of a node $v$ or an edge $e$ (the degree of an edge is the number of other edges sharing an endpoint to it) by $d(v)$ and $d(e)$, respectively. If we have a subset of nodes $U \subseteq V$ or edges $E' \subseteq E$, we will denote $d_U(v)$ to be the number of nodes $u \in U$ such that $\{u,v\} \in E$, and $d_{E'}(v)$ to be the number of edges $e \in E'$ such that $v \in e$. %(we do not require $v$ to be in $U$ to use the notation $d_U(v)$).
We define degree of edges $d_{E'}(e)$ to be the number of edges in $E'$ which are adjacent to $e$. We will use $u \sim v$ to denote adjacency between nodes (or edges), with the underlying graph as a subscript where it is otherwise ambiguous.

Throughout the paper for any positive integer $\ell$, we use $[\ell]$ to denote the set $\{1, \dots, \ell\}$.

\subsection{Luby's MIS algorithm}
\label{subsec-Luby-MIS}

Both our MIS algorithm and our maximal matching algorithm will be based on Luby's algorithm \cite{Luby86} for MIS:

\begin{algorithm}[H]
	\caption{Luby's MIS algorithm}
	\label{alg:Luby}
	\begin{algorithmic}
		\While{$|E(G)|>0$}
		\State Each node $v$ generates a random value $z_v\in [0,1]$
		\State Node $v$ joins independent set $\IS$ iff $z_v< z_u$ for all $u \sim v$
		\State Add $\IS$ to output independent set
		\State Remove $\IS$ and $N(\IS)$ from the graph $G$
		\EndWhile
	\end{algorithmic}
\end{algorithm}

The central idea in the analysis is to define an appropriated subset of nodes and show that it is adjacent to a constant fraction of edges in the graph $G$. Let $X$ be the set of all nodes $v$ that have at least $\frac{d(v)}{3}$ neighbors $u$ with $d(u) \le d(v)$. Then the following lemma is shown, for example, in Lemma 8.1 of \cite{Vigoda06}.

\begin{lemma}
\label{lem:Luby}
Let $X$ be the set of all nodes $v$ that have at least $\frac{d(v)}{3}$ neighbors $u$ with $d(u) \le d(v)$. Then $\sum_{v \in X} d(v) \ge \frac 12 |E|$.
\end{lemma}

Next, one can then show that every node $v \in X$ has a constant probability of being removed from $G$, and therefore, in expectation, a constant fraction of $G$'s edges are removed.

This approach gives an $O(\log n)$-round \emph{randomized} algorithm for MIS (with $\spac = n^{\eps}$). Luby showed, also in \cite{Luby86}, that the analysis requires only pairwise independent random choices, and that the algorithm can thus be efficiently \emph{derandomized} (in $O(\log^3n \log\log n)$ parallel time). However, doing so directly requires many machines ($O(mn^2) = O(n^4)$ in \cite{Luby86}), which would generally be considered a prohibitively high total space bound in \MPC.

The approach used in Luby's MIS algorithm can be also extended to find maximal matching, since a maximal matching in $G$ is an MIS in the line graph of $G$, and in many settings one can simulate Luby's algorithm on this line graph.

\subsection{Communication in low-space \MPC}
\label{subsec:communication-in-fully-scalable-MPC}

Low-space \MPC is in some ways a restrictive model, and even fully scalable algorithms for routing and communication therein are highly non-trivial. Fortunately, prior work on MapReduce and earlier models of parallel computation have provided black-box tools which will permit all of the types of communication we require for our algorithms. We will not go into the details of those tools, but instead refer the reader to the following summary:

\begin{lemma}[\cite{GSZ11}]
\label{lem:comm}
For any positive constant $\eps$, sorting and computing prefix sums of $n$ numbers can be performed deterministically in MapReduce (and therefore in the \MPC model) in a constant number of rounds using $\spac = n^\eps$ space per machine and $O(n)$ total space.
\end{lemma}

The computation of prefix sums here means the following: each machine $m \in [\machines]$ holds an input value $x_m$, and outputs %the value
$\sum_{i=1}^{\machines}x_i$.

\begin{proof}
The result for sorting follows from applying Theorem 3.1 of \cite{GSZ11} to the BSP sorting algorithm of \cite{Goodrich99}. The prefix sums result comes from Lemma 2.2 of \cite{GSZ11}.
\end{proof}

This result essentially allows us to perform all of the communication we will need to do in a constant number of rounds. For example, by sorting edges according to nodes identifiers, we can ensure that the neighborhoods of all nodes are stored on contiguous blocks of machines. Then, by computing prefix sums, we can compute sums of values among a node's neighborhood, or indeed over the whole graph. Where 2-hop neighborhoods fit in the memory of a single machine, we can collect them by sorting edges to collect $1$-hop neighborhoods onto machines, and then having each such machine send requests for the neighborhoods of all the nodes it stores.

\subsection{Families of $k$-wise-independent hash functions}

Our derandomization is based on a classic recipe: we first show~that a randomized process using a \emph{small random seed} produces good results, by using our random seed to select a hash function from a $k$-wise independent family. Then, we search the space of random seeds to find a good one, using the \emph{method of conditional expectations} (sometimes called the \emph{method of conditional probabilities}).

The families of hash functions we require are specified as follows:

\begin{definition}
For $N, L, k \in \nat$ such that $k \le N$, a family of functions $\mathcal{H} = \{h : [N] \rightarrow [L]\}$ is \emph{$k$-wise independent} if for all distinct $x_1, \dots, x_k \in [N]$, the random variables $h(x_1), \dots, h(x_k)$ are independent and uniformly distributed in $[L]$ when $h$ is chosen uniformly at random from~$\mathcal{H}$.
\end{definition}

We will use the following well-known lemma (see, e.g., Corollary~3.34 in \cite{Vadhan12}).

\begin{lemma}
\label{lem:hash}
For every $a$, $b$, $k$, there is a family of $k$-wise independent hash functions $\mathcal{H} = \{h : \{0,1\}^a \rightarrow \{0,1\}^b\}$ such that choosing a random function from $\mathcal{H}$ takes $k \cdot \max\{a,b\}$ random bits, and evaluating a function from $\mathcal{H}$ takes time $poly(a,b,k)$.
\end{lemma}

For all of our purposes (except when extending to low degree inputs, in Section \ref{sec:logD-rounds}), when we require a family of hash functions, we will use a family of $c$-wise independent hash functions $\mathcal{H} = \{h : [n^3] \rightarrow [n^3]\}$, for sufficiently large constant $c$ (we can assume that $n^3$ is a power of $2$ without affecting asymptotic results). We choose $n^3$ to ensure that our functions have (more than) large enough domain and range to provide the random choices for all nodes and edges in our algorithms. By Lemma \ref{lem:hash}, a random function can be chosen from $\mathcal{H}$ using $O(\log n)$ random bits (defining the \emph{seeds}).

\subsection{Method of conditional expectations}
\label{sec:condexp}

Another central tool in derandomization of algorithms we use is the classical \emph{method of conditional expectations}. In our context, we will show that, over the choice of a random hash function $h \in \mathcal{H}$, the expectation of some objective function (which is a sum of functions calculable by individual machines) is at least some value $Q$. That is,
\begin{align*}
    \mathbf{E}_{h \in \mathcal{H}} \big[{q(h):=\sum_{\text{machines } x} q_x(h)}\big]
        &\ge
    Q
    \enspace.
\end{align*}

Since, by the probabilistic method, this implies the existence of a hash function $h^* \in \mathcal{H}$ for which $q(h^*) \ge Q$, then our goal is to find one such $h^* \in \mathcal{H}$ in $O(1)$ \MPC rounds.

We will find the sought hash function $h^*$ by fixing the $O(\log n)$-bit seed defining it (cf. Lemma \ref{lem:hash}), by having all machines agree gradually on chunks of $\log \spac = \Theta(\log n)$ bits at a time. That is, we iteratively extend a fixed prefix of the seed until we have fixed the entire seed. For each chunk, and for each $i$, $1 \le i \le \spac$, each machine calculates $\mathbf{E}_{h \in \mathcal{H}} \left[ q_x(h) | \Xi_i \right]$, where $\Xi_i$ is the event that the random seed specifying $h$ is prefixed by the current fixed prefix, and then followed by $i$. We then sum these values over all machines for each $i$, using Lemma \ref{lem:comm}, obtaining $\mathbf E_{h} \left[ q(h) | \Xi_i \right]$. By the probabilistic method, at least one of these values is at least $Q$. We fix $i$ to be such that this is the case, and continue.

After $O(1)$ iterations, we find the entire seed to define a hash function $h^* \in \mathcal{H}$ such that $q(h^*) \ge Q$. Since each iteration requires only a constant number of \MPC rounds, this process takes only $O(1)$ rounds in total.

\paragraph{Roadmap.}

In Section \ref{sec:max-matching}, we describe an $O(\log n)$-round algorithm for the maximal matching problem by derandomizing Luby's algorithm, via a deterministic graph sparsification technique. Then, in Section \ref{sec:MIS}, we extend our technique to compute also MIS within $O(\log n)$ rounds. This provides $O(\log \Delta)$-round algorithms for maximal matching and MIS in the regime where the maximum degree $\Delta$ is large, i.e., $\Delta=n^{\Omega(\epsilon)}$. In Section \ref{sec:logD-rounds}, we consider the complementary regime where $\log \Delta=o(\log n)$. The main result of this section is an $O(\log\Delta+\log\log n)$-round algorithm in the \MPC model (and in fact, also in the \congc model).

%---------------------------------------------------------------------------------------------------------------------------------------------------------

%---------------------------------------------------------------------------------------------------------------------------------------------------------

%---------------------------------------------------------------------------------------------------------------------------------------------------------

\section{Maximal matching in $O(\log n)$ \MPC rounds}
\label{sec:max-matching}
%---------------------------------------------------------------------------------------------------------------------------------------------------------

%---------------------------------------------------------------------------------------------------------------------------------------------------------

In this section we present a deterministic fully scalable $O(\log n)$-rounds \MPC algorithm for the maximal matching problem. Later, in Section \ref{sec:logD-rounds}, we will extend this algorithm to obtain a round complexity $O(\log \Delta+\log\log n)$, as promised in Theorem \ref{thm:MM-MIS-Delta}; this improves the bound from this section for $\Delta = n^{o(1)}$.

\begin{theorem}
\label{thm:MM}
For any constant $\eps>0$, maximal matching can be found deterministically in the \MPC model in $O(\log n)$ rounds, using $O(n^{\eps})$ space per machine and $O(m+n^{1+\eps})$ total space.
\end{theorem}

The main idea is to derandomize a variant of a maximal matching due to Luby (cf. Section \ref{subsec-Luby-MIS}), which in $O(\log n)$ rounds finds a maximal matching. In each round of Luby's algorithm one selects some matching $M$ and then removes all nodes in $M$ (and hence all edges adjacent to $M$). It is easy to see that after sufficiently many rounds the algorithm finds maximal matching. The central feature of the randomized algorithm is that in expectation, in each single round one will remove a constant fraction of the edges, and hence $O(\log n)$ rounds will suffice in expectation. This is achieved in two steps. One first selects an appropriated subset of nodes and show that it is adjacent to a constant fraction of edges in the graph $G$ (cf. Lemma \ref{lem:Luby}). Then, one shows that every node $v \in X$ has a constant probability of being removed from $G$ (by being incident to the matching $M$ found in a given round), and therefore, in expectation, a constant fraction of $G$'s edges are removed.

In order to derandomize such algorithm, we will show that each single round can be implemented deterministically in a constant number of rounds in the \MPC model so that the same property will be maintained deterministically: in a constant number of rounds one will remove a constant fraction of the edges, and hence $O(\log n)$ rounds will suffice. This is achieved in three steps:
\begin{itemize}
\item select a set of good nodes $B$ which are adjacent to a constant fraction of the edges,
\item then sparsify to $\CE$ the set of edges incident to $B$ to ensure that in each node has at degree $O(n^{\eps/2})$ in $\CE$, and hence a single machine can store its entire 2-hop neighborhood, and
\item then find a matching $M \subseteq \CE$ such that removal of all nodes in $M$ (i.e., removal of $M$ and all edges adjacent to $M$) reduces the number of edges by a constant factor.
\end{itemize}

\paragraph{Good nodes.}

We start with a corollary of Lemma \ref{lem:Luby}, which specifies a set of \emph{good nodes} which are nodes with similar degrees that are adjacent to a constant fraction of edges in the graph.

Let $\delta$ be an arbitrarily positive constant, $1/\delta \in \nat$. We will proceed in a constant (dependent on $\delta$) number of stages, sparsifying the graph induced by the edges incident to good nodes by derandomizing the sampling of edges with probability $n^{-\delta}$ in each stage. In order for this to work, we want our good nodes to be within a degree range of at most a $n^\delta$ factor, for their behavior to be similar.

Let us recall (cf. Section \ref{subsec-Luby-MIS}) that $X$ is the set of all nodes $v$ which have at least $\frac{d(v)}{3}$ neighbors $u$ with $d(u) \le d(v)$. Partition nodes into sets $C^i$, $1 \le i \le 1/\delta$, such that $C^i = \{v: n^{(i-1)\delta} \le d(v)< n^{i\delta}\}$. Let $B^i = C^i \cap X$. The following is a simple corollary of Lemma~\ref{lem:Luby}.

\begin{corollary}
\label{cor:MMpart}
There is $i \le 1/\delta$, such that $\sum_{v \in B_i} d(v) \ge \frac{\delta}{2}|E|$.
\end{corollary}

\begin{proof}
By Lemma \ref{lem:Luby}, $\sum_{v \in X} d(v) \ge |E|$. Since the sets $B^1, \dots, B^{1/\delta}$ form a partition of $X$ into $1/\delta$ subsets, at least one of them must contribute a $\delta$-fraction of the sum $\sum_{v\in X} d(v) \ge \frac12 |E|$.
\end{proof}

From now on, let us fix some $i$ which satisfies Corollary \ref{cor:MMpart}. Denote $B := B^i$, and for each node $v \in B$, let $X(v) := \{\{u,v\} \in E : d(u) \le d(v)\}$. Note that the definition of set $X$ yields $|X(v)| \ge \frac{d(v)}{3}$.
Denote $\E_0 = \bigcup_{v \in B} X(v)$.
$\E_0$ is the set of edges we will be sub-sampling to eventually find a matching, and $B$ is the set of good nodes which we want to match and remove from the graph, in order to significantly reduce the number of edges.

The outline of our maximal matching algorithm is as follows:

\begin{algorithm}[H]
	\caption{Maximal matching algorithm outline}
	\label{alg:MM}
	\begin{algorithmic}
		\While{$|E(G)|>0$}
		\State Compute $i$, $B$ and $\E_0$
		\State Select a set $\CE \subseteq \E_0$ that induces a low degree subgraph
		\State \mbox{Find matching $M \subseteq \CE$ with $\sum_{\text{matched nodes } v}d(v) = \Omega(|E(G)|)$}
%		\State Find matching $M \subseteq \CE$ with $\sum\limits_{\text{matched nodes } v}d(v) = \Omega(|E(G)|)$
		\State Add $M$ to output matching, remove matched nodes from $G$
		\EndWhile
	\end{algorithmic}
\end{algorithm}

As long as each iteration reduces the number of edges in $G$ by a constant fraction, we get only $O(\log n)$ iterations to find a maximal matching. We will show that the iterations require a constant number of rounds each, so $O(\log n)$ rounds are required overall.

\subsection{Computing $i$, $B$, and $\E_0$}
\label{subsec:comput-i-B-E}

As discussed in Section \ref{subsec:communication-in-fully-scalable-MPC}, a straightforward application of Lemma \ref{lem:comm} allows all nodes to determine their degrees, and therefore their membership of sets $C^i$, in a constant number of rounds. A second application allows nodes to determine whether they are a member of $X$, and therefore $B^i$, and also provides nodes $v \in X$ with $X(v)$. Finally, a third application allows the computation of the values $\sum_{v \in B^i} d(v)$ for all $i$. Upon completing, all nodes know which $i$ yields the highest value for this sum, and that is the value for $i$ which will be fixed for the remainder of the algorithm.

\subsection{Deterministically selecting $\CE$ that induces a low degree subgraph}
\label{subsec:E^*}

We will show now how to deterministically, in $O(1)$ stages, select a subset $\CE$ of $\E_0$ that induces a low degree subgraph, as required in our \MPC algorithm. For that, our main goal is to ensure that every node has degree $O(n^{4\delta})$ in $\CE$ (to guarantee that its 2-hop neighborhood will fit a single \MPC machine with $\spac = O(n^{8\delta})$), and that one can then locally find a matching $M \subseteq \CE$ that will cover a linear number of edges.

We first consider the easy case when $i \le 4$, in which case we set directly $\CE = \E_0$. Notice that in that case, by definitions of $X$ and $B = C^i \cap X$, we have
\begin{inparaenum}[(i)]
\item $d_{\CE}(v) = d_{\E_0}(v) \le n^{4 \delta}$ for all nodes $v$, and
\item $|X(v) \cap \CE| = |X(v)| \ge \frac{d(v)}{3}$ for all nodes $v \in B$,
\end{inparaenum}
which is what yields the requirements from $\CE$ (cf. the Invariant below) needed in our analysis in Section \ref{subsec:CE->M}.

Next, for the rest of the analysis, let us assume that $i \ge 5$. We proceed in $i-4$ stages, starting with $\E_0$ and sparsifying it by sub-sampling a new edge set $\E_j$ in each stage $j$, $j = 1, 2, \dots, i-4$. Note that for any node $v$ we have $d_{\E_0}(v) \le n^{i\delta}$, since nodes in $B$ have maximum degree $n^{i\delta}$ and since $v$ only has adjacent edges in $\E_0$ if $v \in B$ or $\exists u \in B: d(v) \le d(u)$.

\paragraph{Invariant:}

In our construction of sets $\E_0, \E_1, \dots, \E_{i-4}$, in order to find a good matching in the resulting sub-sampled graph $\CE$, we will maintain the following invariant for every $j$:
\begin{enumerate}[(i)]
\item for all nodes $v$: $d_{\E_j}(v) \le (1+o(1)) n^{-j\delta} d_{\E_0}(v) + n^{3\delta}$,
\item for all nodes $v \in B$: $|X(v) \cap \E_j| \ge (1-o(1)) n^{-\delta j} |X(v)|$.
\end{enumerate}

The intuition behind this invariant is that nodes' degrees decrease roughly as expected in the sub-sampled graph, and nodes $v \in B$ do not lose too many edges to their neighbors in $X(v)$ (to ensure that many of them can be matched in the sub-sampled graph).

One can see that the invariant holds for $j = 0$ trivially, by definition of sets $\E_0$ and $B$.

\paragraph{Distributing edges and nodes among the machines.}

In order to implement our scheme in the \MPC model, we first allocate the nodes and the edges of the graph among the machines.
\begin{itemize}
\item Each node $v$ distributes its adjacent edges in $\E_{j-1}$ across a group of \emph{type $A$} machines, with $n^{4 \delta}$ edges on all but at most one machine (which holds any remaining edges).
\item Each node $v\in B$ also distributes its adjacent edges in $X(v) \cap \E_{j-1}$ across a group of \emph{type $B$} machines in the same fashion.
\end{itemize}

Type A machines will be used to ensure that the first point of the invariant holds, and type B machines will ensure the second.

In order to sparsify $\E_{j-1}$ to define $\E_j$, we proceed with derandomization of a sub-sampled graph. We will fix a seed specifying a hash function from $\mathcal{H}$ (recall that $\mathcal{H} = \{h : [n^3] \rightarrow [n^3]\}$ is a $c$-independent family for sufficiently large constant $c$). Each hash function $h$ induces a set $\E_h$ in which each edge in $\E_{j-1}$ is sampled with probability $n^{-\delta}$, by placing $e \in \E_h$ iff $h(e) \le n^{3-\delta}$.

\paragraph{Good machines.}

We will call a machine \emph{good} for a hash function $h \in \mathcal{H}$ if the effect of $h$ on the edges it stores looks like it will preserve the invariant. We will then show that if all machines are good for a hash function $h$, the invariant is indeed preserved.

Formally, consider a machine (of either type) $x$ that receives $\E(x) \subseteq \E_{j-1}$ and let $e_x := |\E(x)|$. For hash function $h \in \mathcal{H}$, we call $x$ \emph{good} if $e_x n^{-\delta} - n^{0.1\delta} \sqrt{e_x} \le |\E(x) \cap \E_h| \le e_x n^{-\delta} + n^{0.1\delta} \sqrt{e_x}$.

Our aim is to use the following concentration bound to show that a machine is good with high probability:

\begin{lemma}[Lemma 2.2 of \cite{BR94}]
\label{lem:conc}
Let $\cj \ge 4$ be an even integer. Let $Z_1, \dots, Z_t$ be $\cj$-wise independent random variables taking values in $[0,1]$, $Z = Z_1 + \dots + Z_t$ and $\mu = \Exp{Z}$. Let $\lambda > 0$. Then,
\begin{align*}
	\Prob{|Z - \mu| \ge \lambda} &\le 2 \left(\frac{\cj t}{\lambda^2}\right)^{\cj/2}
	\enspace.
\end{align*}
\end{lemma}

We will take $Z$ to be the sum of the indicator variables $\mathbf{1}_{\{e\in \E_h\}}$ for $e \in \E(x)$ (i.e., $Z = |\E(x) \cap \E_h|$). These indicator variables $\mathbf{1}_{\{e \in \E_h\}}$ are $\cj$-wise independent, and each has expectation $n^{-\delta}$. Using that $\cj$ is a sufficiently large constant, we apply Lemma \ref{lem:conc} and get that
\begin{align*}
    \Prob{|Z - \mu| \ge n^{0.1 \delta} \sqrt{e_x}} &\le n^{-5}
    \enspace.
\end{align*}

This means that with high probability, $e_x n^{-\delta} - n^{0.1 \delta} \sqrt{e_x } \le |\E(x) \cap \E_h| \le e_x n^{-\delta} + n^{0.1 \delta} \sqrt{e_x}$, and $x$ is good.

By the method of conditional expectations, as described in Section \ref{sec:condexp} using objective function $q_x(h) = \textbf{1}_{\text{$x$ is good for $h$}}$, we can find a function $h$ which makes all machines good, in a constant number of rounds. We then set $\E_{j} = \E_h$.

\subsubsection{Properties of $\E_j$: satisfying the invariant}
\label{subsubsec:Properties-of-Ej}

Having fixed a sub-sampled graph for the stage, we need to show that since all machines were good, we satisfy our invariant for the stage.\junk{ We defer the proofs of the following Lemmas \ref{lemma-inv-i}--\ref{lemma-inv-ii} to Appendix \ref{app:matching}.}

\begin{lemma}[Invariant (i)]
\label{lemma-inv-i}
All nodes $v$ satisfy
\begin{align*}
    d_{\E_j}(v) \le (1+o(1)) n^{-j\delta} d_{\E_0}(v) + n^{3\delta}\enspace.
\end{align*}
\end{lemma}
\junk{
\begin{lemma}%[Invariant \emph{(i)}]
\label{lemma-inv-i}
All nodes $v$ satisfy \mbox{$d_{\E_j}(v) \le (1+o(1)) n^{-j\delta} d_{\E_0}(v) + n^{3\delta}$.}
\end{lemma}
}

\def\APPENDMATCHINGINVONE{
\begin{proof}%[Proof of Lemma \ref{lemma-inv-i}]
Node $v$'s adjacent edges in $\E_{j-1}$ were divided among $\lfloor\frac{d_{\E_{j-1}}(v)}{n^{3\delta}}\rfloor$ type $A$ machines containing $n^{4\delta}$ neighbors, and one machine containing the remaining $d_{\E_{j-1}}(v)-n^{4\delta}\lfloor\frac{d_{\E_{j-1}}(v)}{n^{4\delta}}\rfloor\le d_{\E_{j-1}}(v)$ neighbors. Therefore we obtain:
	\begin{align*}
	d_{\E_j}(v) &\le
	\sum_{\substack{\text{$v$'s type $A$}\\\text{machines $x$}}}
	v_x n^{-\delta} + n^{0.1\delta} \sqrt{v_x}
	\\
	&\le
	n^{-\delta} d_{\E_{j-1}}(v) +
	\left\lfloor\frac{d_{\E_{j-1}}(v)}{n^{4\delta}}\right\rfloor
	n^{0.1\delta} \sqrt{n^{4\delta}} +
	n^{0.1\delta} \sqrt{d_{\E_{j-1}}(v)}
	\\
	&\le
	n^{-\delta} d_{\E_{j-1}}(v) +
	\frac{n^{0.1\delta}}{n^{2\delta}} d_{\E_{j-1}}(v) +
	n^{0.1\delta} \sqrt{d_{\E_{j-1}}(v)}
	\enspace.
	\end{align*}
	
	If $d_{\E_{j-1}}(v)\ge n^{3\delta}$, we have
	\begin{align*}
	d_{\E_j}(v) &\le
	n^{-\delta} d_{\E_{j-1}}(v) +
	\frac{n^{0.1\delta}}{n^{2\delta}} d_{\E_{j-1}}(v) +
	n^{0.1\delta} \sqrt{d_{\E_{j-1}}(v)}
	\\
	&\le
	n^{-\delta} d_{\E_{j-1}}(v) +
	n^{-1.9\delta} d_{\E_{j-1}}(v) +
	n^{-1.4\delta} d_{\E_{j-1}}(v)
	\\
	&=
	(1+o(1))n^{-\delta} d_{\E_{j-1}}(v)
	\\
	&\le
	(1+o(1)) n^{-j\delta} d_{\E_0}(v)
	\enspace.
	\end{align*}
	
Otherwise, $d_{\E_{j}}(v)\le d_{\E_{j-1}}(v)\le n^{3\delta}$. In either case, we satisfy the invariant for stage $j$.
\end{proof}
}
\APPENDMATCHINGINVONE

\begin{lemma}[Invariant (ii)]
\label{lemma-inv-ii}
All nodes $v \in B$ satisfy
\begin{align*}
    |X(v) \cap \E_j| \ge (1-o(1)) n^{-\delta j} |X(v)|\enspace.
\end{align*}
\end{lemma}
\junk{
\begin{lemma}%[Invariant \emph{(ii)}]
\label{lemma-inv-ii}
All nodes $v \in B$ satisfy $|X(v) \cap \E_j| \ge (1-o(1)) n^{-\delta j} |X(v)|$.
\end{lemma}
}

\def\APPENDMATCHINGINVTWO{
\begin{proof}%[Proof of Lemma \ref{lemma-inv-ii}]
	The edges in $X(v) \cap \E_{j-1}$ were divided among $\lfloor\frac{|X(v) \cap \E_{j-1}|}{n^{3\delta}}\rfloor$ type $B$ machines containing $n^{4\delta}$ neighbors, and one machine containing the remaining $|X(v) \cap \E_{j-1}| - n^{4\delta}\lfloor\frac{|X(v) \cap \E_{j-1}|}{n^{4\delta}}\rfloor \le |X(v) \cap \E_{j-1}|$ neighbors.
	\begin{align*}
	|X(v) \cap \E_j|
	&\ge
	\sum_{\substack{\text{$v$'s type $B$}\\\text{machines $x$}}}
	v_x n^{-\delta}-n^{0.1\delta}\sqrt{v_x }
	\\
	&=
	n^{-\delta}|X(v)\cap \E_{j-1}| -
	\left\lfloor\frac{|X(v) \cap \E_{j-1}|}{n^{4\delta}}\right\rfloor
	n^{0.1\delta} \sqrt{n^{4\delta}} -
	n^{0.1\delta} \sqrt{|X(v) \cap \E_{j-1}|}
	\\
	&\ge
	n^{-\delta}|X(v)\cap \E_{j-1}|-\frac{n^{0.1\delta}}{n^{2\delta}}|X(v)\cap \E_{j-1}|
	- n^{0.1\delta}\sqrt{|X(v)\cap \E_{j-1}|}
	\enspace.
	\end{align*}
	
	Since, by the invariant for stage $j-1$, we have $|X(v)\cap \E_{j-1}|\ge n^{4\delta}$, we can continue the lower bound from above to obtain,
	\begin{align*}
	|X(v) \cap \E_j| &\ge
	n^{-\delta}|X(v) \cap \E_{j-1}| -
	\frac{n^{0.1 \delta}}{n^{2 \delta}}|X(v) \cap \E_{j-1}| -
	n^{0.1 \delta} \sqrt{|X(v) \cap \E_{j-1}|}
	%        \\
	%        &\ge
	%    n^{-\delta}|X(v) \cap \E_{j-1}| -
	%        \frac{n^{0.1 \delta}}{n^{2 \delta}}n^{4\delta} -
	%        n^{0.1 \delta} n^{2\delta}
	%        \\
	%        &=
	%    n^{-\delta}|X(v) \cap \E_{j-1}| - 2 n^{2.1 \delta}
	\\
	&\ge
	n^{-\delta}|X(v) \cap \E_{j-1}|- n^{3\delta}
	\\
	&=
	(1-o(1)) n^{-\delta} |X(v) \cap \E_{j-1}|
	\\
	&\ge
	(1-o(1) )n^{-j\delta} |X(v)|
	\enspace,
	\end{align*}
	where the last inequality follows directly from Invariant (ii).
\end{proof}
}
\APPENDMATCHINGINVTWO
We have proven that our invariant is preserved in every stage, and therefore holds in our final sub-sampled edge set $\CE := \E_{i-4}$.

\subsection{Finding a matching $M \subseteq \CE$}
\label{subsec:CE->M}

The construction in Section \ref{subsec:E^*} ensures that either $i \le 4$, in which case $\CE = \E_0$, or $i \ge 5$ and after $i-4$ stages, we now have a set of edges $\CE = \E_{i-4}$ with the following properties:
\begin{enumerate}[(i)]
\item all nodes $v$ have $d_{\CE}(v) \le (1+o(1)) n^{(4-i)\delta} d_{\E_0}(v) + n^{3\delta} \le 2n^{4\delta}$,
    \begin{align*}
        d_{\CE}(v) &\le
        (1+o(1)) n^{(4-i)\delta} d_{\E_0}(v) + n^{3\delta} \le 2n^{4\delta}
        \enspace,
    \end{align*}
\item all nodes $v \in B$ have
    \begin{align*}
        |X(v) \cap \CE| \ge (1-o(1)) n^{(4-i)\delta} |X(v)|
        \enspace.
    \end{align*}
\end{enumerate}
\junk{
\begin{enumerate}[(i)]
\item for all nodes $v$: $d_{\CE}(v) \le (1+o(1)) n^{(4-i)\delta} d_{\E_0}(v) + n^{3\delta} \le 2n^{4\delta}$,
\item for all nodes $v \in B$: $|X(v) \cap \CE| \ge (1-o(1)) n^{(4-i)\delta} |X(v)|$.
\end{enumerate}
}

We can show a property analogous to Lemma \ref{lem:Luby} in \CE \junk{(for a proof, see Appendix \ref{app:egoodsample})}\footnote{Recall (cf. Section \ref{sec:notation}) that a degree of an edge
$e$ in \CE, $d_{\CE}(e)$, is the number of edges in \CE adjacent to~$e$.}.

\begin{lemma}
\label{lem:egoodsample}
Every node $v \in B$ either satisfies $\sum_{\{u,v\} \in \CE} \frac{1}{d_{\CE}(\{u,v\})} \ge \frac{1}{27}$, or is incident to an edge $\{u,v\} \in \CE$ whose degree in \CE is $0$.
\end{lemma}

\def\APPENDMATCHINGGOODSAM{
\begin{proof}%[Proof of Lemma \ref{lem:egoodsample}]
Fix $v \in B$. We begin with the case $i \ge 5$, and then proceed to the simpler case $i \le 4$.
	
\paragraph{Case $i \ge 5$:}
If $v$ does not have an incident edge $\{u,v\} \in \CE$ whose degree in $\CE$ is $0$, then:
\begin{align*}
	\sum_{\{u,v\} \in \CE} \frac{1}{d_{\CE}(\{u,v\})}
	&\ge
	\sum_{\{u,v\} \in \CE \cap X(v)} \frac{1}{d_{\CE}(\{u,v\})}
	\ge
	\sum_{\{u,v\} \in \CE \cap X(v)} \frac{1}{d_{\CE}(u) + d_{\CE}(v)}\\
	&\ge
	\sum_{\{u,v\} \in \CE \cap X(v)} \frac{1}{2n^{(4-i)\delta} (d(u)+d(v)) + 2n^{3\delta}}\\
	&\ge
	\sum_{\{u,v\} \in \CE \cap X(v)} \frac{1}{2n^{(4-i)\delta}(2d(v) + n^{(i-1)\delta})}\\
	&=
	\frac{|\CE \cap X(v)|}{2n^{(4-i)\delta}(2d(v) + n^{(i-1)\delta})}
	\ge
	\frac{n^{(4-i)\delta} |X(v)|}{3n^{(4-i)\delta} (2d(v)+n^{(i-1)\delta})}
	\\
	&\ge
	\frac{|X(v)|}{9d(v)}
	\ge
	\frac{1}{27}
	\enspace.
\end{align*}
Here, the 3rd inequality follows from Invariant \emph{(i)}, the 4th inequality follows from the fact that $\{u,v\} \in X(v)$ yields $d(u) \le d(v)$, the 5th inequality follows from Invariant \emph{(ii)}, the 6th inequality follows from the fact that $v \in C^i$ and hence $d(v) \ge n^{(i-1) \delta}$, and the last inequality follows from the fact that $v \in B$ yields $v \in X$ and hence $|X(v)| \ge \frac{d(v)}3$.

%\medskip

\paragraph{Case $1 \le i \le 4$:}
The proof for the case $1 \le i \le 4$ uses the fact $\CE = \E_0$, and hence
\begin{inparaenum}[\it (i)]
\item $d_{\CE}(v) = d_{\E_0}(v)$ for all nodes $v$, and
\item $|X(v) \cap \CE| = |X(v)| \ge \frac{d(v)}{3}$ for all nodes $v \in B$.
\end{inparaenum}
For a fixed $v \in B$, if $v$ does not have an incident edge $\{u,v\} \in \CE$ whose degree in $\CE$ is $0$, then:
\begin{align*}
	\sum_{\{u,v\} \in \CE} \frac{1}{d_{\CE}(\{u,v\})}
	&=
	\sum_{\{u,v\} \in \CE \cap X(v)} \frac{1}{d_{\CE}(\{u,v\})}
    \ge
	\sum_{\{u,v\} \in \CE \cap X(v)} \frac{1}{d_{\CE}(u) + d_{\CE}(v)}\\
	&\ge
	\sum_{\{u,v\} \in \CE \cap X(v)} \frac{1}{d(u)+d(v)}
    \ge
	\sum_{\{u,v\} \in \CE \cap X(v)} \frac{1}{2d(v)}\\
	&=
	\sum_{\{u,v\} \in X(v)} \frac{1}{2d(v)}
	=
	\frac{|X(v)|}{2d(v)}
	\ge
	\frac{1}{6}
	\enspace.
\end{align*}
Here, in the third inequality we use the fact that $\{u,v\} \in X(v)$ yields $d(u) \le d(v)$, and the last inequality follows from the fact that $v \in B$ yields $v \in X$ and hence $|X(v)| \ge \frac{d(v)}3$.
\end{proof}
}
\APPENDMATCHINGGOODSAM
Now we are ready to present our deterministic \MPC algorithm that for a given subset of edges $\CE$ satisfying the invariant, in $O(1)$ rounds constructs a matching $M \subseteq \CE$ such that the removal of $M$ and all edges adjacent to $M$ removes $\Omega(\delta |E|)$ edges from the graph.

First, each node $v \in B$ is assigned a machine $x_v$ which gathers its 2-hop neighborhood in $\CE$. Since for every node $u$ we have $d_{\CE}(u) \le 2n^{4\delta}$ by Invariant \emph{(i)} (or by the definition of $B$ and $\E_0 = \CE$, when $1 \le i \le 4$), this requires at most $2n^{4\delta}\cdot 2n^{4\delta} = O(n^{8\delta})$~space per machine. Altogether, since $|B| \le n$, this is $O(n^{1+8\delta})$ total space.

We will fix a seed specifying a hash function $h$ from $\mathcal{H}$. This hash function $h$ will be used to map each edge $e$ in $\CE$ to a value $z_e \in [n^{3}]$. Then, $e$ joins the \emph{candidate matching $\E_h$} iff $z_e < z_{e'}$ for all $e' \sim e$. Further, since each node $v \in B$ is assigned a machine which gathers its 2-hop neighborhood in $\CE$, in a single \MPC round, every node $v \in B$ can determine its degree $d_{\E_h}(v)$.

Clearly $\E_h$ is indeed a matching for every $h \in \mathcal{H}$, but we require that removing $\E_h \cup N(\E_h)$ from the graph reduces the number of edges by a constant fraction. We will show that $|\E_h \cup N(\E_h)| = \Omega(\delta |E|)$ in expectation, and therefore by the method of conditional expectations (cf. Section \ref{sec:condexp}) we will be able to find a seed $h^* \in \mathcal{H}$ for which $|\E_{h^*} \cup N(\E_{h^*})| = \Omega(\delta |E|)$.

\junk{Each machine $x_v$ is \emph{good} for a hash function $h$ if it holds an edge $\{u,v\}\in \E_h$. Since $x_v$ holds the 2-hop neighborhood of $v$ in $\CE$, it can determine whether $v$'s adjacent edges in $\CE$ are members of $\E_h$. We show that with constant probability, $x_v$ is good for a random hash function~$h$.}
\junk{
The proof of the following lemma is deferred to Appendix \ref{app:TwoHop}.}

\begin{lemma}
\label{lem:egoodhash}
For any machine $x_v$ holding the 2-hop neighborhood of $v$ in $\CE$, the probability that $d_{\E_h}(v)=1$, for a random hash function $h \in \mathcal{H}$, is at least $\frac{1}{218}$.
\end{lemma}

\def\APPENDMATCHINGTWOHOP{
\begin{proof}%[Proof of Lemma \ref{lem:egoodhash}]
If $v$ has an adjacent edge $e$ with degree $0$, then $e$ will join $\E_h$ and we are done.
	
Otherwise, for any edge $\{u,v\} \in \CE$ it holds that,
\begin{align*}
	\frac{1}{2d_{\CE}(\{u,v\})}
	&\ge
	\Prob{z_{\{u,v\}} < \frac{n^3}{2d_{\CE}(\{u,v\})}}
	\ge
	\frac{1}{2d_{\CE}(\{u,v\})} - \frac {1}{n^3}
	\enspace.
\end{align*}
Conditioned on $z_{\{u,v\}} < \frac{n^3}{2d_{\CE}(\{u,v\})}$, the probability that $\{u,v\} \in \E_h$ is at least
\begin{align*}
	\Prob{\{u,v\} \in \E_h \Big| z_{\{u,v\}} < \tfrac{n^3}{2d_{\CE}(\{u,v\})}}
    &=
	\Prob{\forall_{e \in \CE \sim \{u,v\}} \ z_{\{u,v\}} < z_e \Big|
		z_{\{u,v\}} < \tfrac{n^3}{2d_{\CE}(\{u,v\})}}
	\\
    &=
	1 - \Prob{\exists_{e \in \CE \sim \{u,v\}} \ z_e \le z_{\{u,v\}} \Big|
		z_{\{u,v\}} < \tfrac{n^3}{2d_{\CE}(\{u,v\})}}
	\\
	&\ge
	1 - \sum_{e \in \CE \sim \{u,v\}}
	\Prob{z_e \le z_{\{u,v\}} \Big| z_{\{u,v\}} < \tfrac{n^3}{2d_{\CE}(\{u,v\})}}
	\\
	&\ge
	1 - \sum_{e \in \CE \sim \{u,v\}}
	\Prob{z_e < \tfrac{n^3}{2d_{\CE}(\{u,v\})}}
	\\
	&\ge
	1 - d_{\CE}(\{u,v\}) \cdot \frac{1}{2d_{\CE}(\{u,v\})}
	\\
	&=
	\frac12
	\enspace,
\end{align*}
by pairwise independence of the family of hash functions $\mathcal{H}$.
	
Therefore, the probability that $d_{\E_h}(v) = 1$ is at least:
\begin{align*}
	\Prob{d_{\E_h}(v)=1}
	&=
	\sum_{u \sim_\CE v} \Prob{\{u,v\} \in \E_h}
	\ge
	\sum_{u \sim_\CE v} \Prob{\{u,v\} \in \E_h \land
		z_{\{u,v\}} < \tfrac{n^3}{2d_{\CE}(\{u,v\})}}
	\\
	&\ge
	\sum_{u \sim_\CE v} \Prob{\{u,v\} \in \E_h \Big|
		z_{\{u,v\}} < \tfrac{n^3}{2d_{\CE}(\{u,v\})}}
	\cdot
	\Prob{z_{\{u,v\}} < \tfrac{n^3}{2d_{\CE}(\{u,v\})}}
	\\
	&\ge
	\sum_{u \sim_\CE v} \frac12 \cdot
	\left(\frac{1}{2d_{\CE}(\{u,v\})} - \frac {1}{n^3}\right)
	\ge
	\frac14 \cdot \sum_{u \sim_\CE v} \frac{1}{d_{\CE}(\{u,v\})}
	- \frac{1}{2n^2}
	\ge
	\frac {1}{108} - \frac{1}{2n^2}
	%        \Prob{\exists\{u,v\} \in \CE: z_{\{u,v\}} \le \tfrac{n^3}{2 d_{\CE}({\{u,v\}})}}
	\enspace,
\end{align*}
where the first identity follows since the events are mutually exclusive, and the last one follows by Lemma~\ref{lem:egoodsample}. The claim now follows from $\frac {1}{108} - \frac{1}{2n^2} \ge \frac{1}{109}$ for large enough $n$.
\end{proof}
}
\APPENDMATCHINGTWOHOP

We will denote $N_h := \{v \in B: d_{\E_h}(v) = 1\}$, i.e., the set of nodes in the matching induced by hash function $h$. We want to study the number of edges incident to $N_h$. By Lemma \ref{lem:egoodhash},
\begin{align*}
    \Exp{\sum_{v \in N_h} d(v)}
        &\ge
    \sum_{v \in B} d(v) \cdot \Prob{v \in N_h}
        \ge
    \frac{1}{109} \sum_{v \in B}d(v)
        \ge
    \frac{\delta |E|}{218}
    \enspace.
\end{align*}

By the method of conditional expectations (cf. Section \ref{sec:condexp}), using objective function $q_{x_v}(h) = d(v) \textbf 1_{v \in N_h}$, we can select a hash function $h$ with $\sum_{v \in N_h} d(v) \ge \frac{1}{109} \delta |E| $. We then add the matching $M := \E_h$ to our output, and remove matched nodes from the graph. In doing so, we remove at least $\frac{\delta |E|}{536}$ edges from the graph.

\subsection{Completing the proof of Theorem \ref{thm:MM}: finding a maximal matching}

Now we are ready to complete the proof of Theorem \ref{thm:MM}, that a maximal matching can be found deterministically in the \MPC model in $O(\log n)$ rounds, with $\spac = O(n^{\eps})$, and $O(m+n^{1+\eps})$ total space.

Our algorithm returns a maximal matching in $\log_{\frac{1}{1-\delta/536}}|E| = O(\log n)$ iterations, each requiring $O(1)$ \MPC rounds. The space required is dominated by storing the input graph $G$ ($O(m)$ total space) and collecting 2-hop neighborhoods when finding an matching ($O(n^{8\delta})$ space per machine, $O(n^{1+8\delta})$ total space). Setting $\delta = \frac{\eps}{8}$ allows us to conclude Theorem \ref{thm:MM}, that for any constant $\eps > 0$, maximal matching can be found in the \MPC model in $O(\log n)$ rounds, using $O(n^{\eps})$ space per machine and $O(m+n^{1+\eps})$ total space.
\qed

%---------------------------------------------------------------------------------------------------------------------------------------------------------

%---------------------------------------------------------------------------------------------------------------------------------------------------------

%---------------------------------------------------------------------------------------------------------------------------------------------------------

\section{Maximal independent set in $O(\log n)$ \MPC rounds}
\label{sec:MIS}
%---------------------------------------------------------------------------------------------------------------------------------------------------------

%---------------------------------------------------------------------------------------------------------------------------------------------------------

In this section we modify the approach from Section \ref{sec:max-matching} for the maximal independent set problem and prove the following.

%In this section we show how to modify the approach from Section \ref{sec:max-matching} for MIS. We will start by considering the case where the maximal degree $\Delta=n^{\Omega(\epsilon)}$ by presenting a deterministic fully scalable \MPC algorithm with $O(\log \Delta)=O(\log n)$ rounds. We will later show how to provide an improved solution of $O(\log \Delta+\log\log n)$ rounds for smaller values of $\Delta$. The bulk of this section is then devoted to showing the following:

\begin{theorem}
\label{thm:MIS}
For any constant $\eps>0$, MIS can be found deterministically in \MPC in $O(\log n)$ rounds, using $O(n^{\eps})$ space per machine and $O(m+n^{1+\eps})$ total space.
\end{theorem}

Later, in Section \ref{sec:logD-rounds}, we will extend this algorithm to obtain a round complexity $O(\log \Delta+\log\log n)$; this will improve the bound from Theorem \ref{thm:MIS} when $\Delta = n^{o(1)}$.

\subsection{Outline}

The approach to find an MIS in $O(\log n)$ \MPC rounds is similar to the algorithm for maximal matching. However, the main difference is that for MIS, instead of the edges, as for the matching, we have to collect the nodes, which happen to require some changes in our analysis and makes some of its part slightly more complex.

Let $\A$ be the set of all nodes $v$ such that $\sum_{u \sim v} \frac{1}{d(u)} \ge \frac13$. Our analysis again relies on a corollary to the analysis of Luby's algorithm (cf. Lemma~\ref{lem:Luby}) that follows from the fact that $X \subseteq \A$.

\begin{corollary}\label{cor:MISpart} %ARTUR: I COMMENTED THE PROOF ...
$\sum_{v \in \A} d(v) \ge \frac 12|E|$.
\end{corollary}

\begin{proof}
Nodes $v$ in $X$ (cf. Lemma \ref{lem:Luby}) satisfy $\sum_{u \sim v} \frac{1}{d(u)} \ge \frac13$.
\end{proof}

\hide{(Here, and in similar sums, we will use the convention that if a node $u$ has $d(u)=0$, then $\frac{1}{d(u)}=+\infty$, and so the value of a sum including such a term is indeed at least $\frac 13$.)}

We will again partition nodes into classes of similar degree. Let $\delta$ be an arbitrarily small constant and assume $1/\delta \in \nat$. As in Section \ref{sec:max-matching}, partition nodes into sets $C^i$, $1 \le i \le 1/\delta$, with $C^i = \{v : n^{(i-1)\delta} \le d(v) < n^{i\delta}\}$. For any $1 \le i \le 1/\delta$, let $\B_i$ be the set of all nodes $v$ satisfying
\junk{
\begin{align*}
    \sum_{u \in C^i \sim v} \frac{1}{d(u)} &\ge \frac{\delta}{3}
    \enspace.
\end{align*}
}
$\sum_{u \in C^i: u \sim v} \frac{1}{d(u)} \ge \frac{\delta}{3}$.
We can easily prove the following.

\begin{corollary}\label{cor:part} %ARTUR: I COMMENTED THE PROOF ...
There is $i \le 1/\delta$, such that $\sum_{v \in \B_i} d(v) \ge \frac\delta 2 |E|$.
\end{corollary}

\begin{proof}
Each element $v \in \A$ must be a member of at least one of the sets $\B_i$, since
\begin{align*}
	\sum_{1 \le i \le 1/\delta} \sum_{u \in C^i \sim v} \frac{1}{d(u)}
	       &=
	\sum_{u \sim v} \frac{1}{d(u)}
        \ge
	\frac13
	\enspace.
\end{align*}
Therefore, there is at least one set $\B_i$ that contributes at least a $\delta$-fraction of the sum $\sum_{v \in \A} d(v)$, i.e., $\sum_{v \in \B_i} d(v) \ge \frac\delta 2 |E|$.
\end{proof}

Henceforth we will fix $i$ to be a value satisfying Corollary \ref{cor:part}, and let $\B := \B_i$ and $\J_0 := C^i$. With this notation, we are now ready to present the outline of our algorithm:

\begin{algorithm}[H]
	\caption{MIS algorithm outline}
	\label{alg:MIS}
	\begin{algorithmic}
		\While{$|E(G)|>0$}
%		\State Add all isolated nodes to output independent set $\MIS$, remove them from $G$
		\State Add all isolated nodes to $\MIS$; remove them from $G$
		\State Compute $i$, $\B$ and $\J_0$
		\State Select a set $\J' \subseteq \J_0$ that induces a low degree subgraph
		\State \mbox{Find independent set $\IS \subseteq \J'$ with $\sum_{v \in N(\IS)} d(v) = \Omega(|E(G)|)$}
%		\State Find independent set $\IS \subseteq \J'$ with $\sum\limits_{v \in N(\IS)} d(v) = \Omega(|E(G)|)$
		\State Add $\IS$ to $\MIS$; remove $\IS$ and $N(\IS)$ from $G$
		\EndWhile
	\end{algorithmic}
\end{algorithm}

Notice that since in each round we find an independent set $\IS$ with $\sum_{v \in N(\IS)} d(v) = \Omega(|E(G)|)$, it is easy to see that Algorithm~\ref{alg:MIS} finds an MIS in $O(\log n)$ rounds. Hence our goal is to find an independent set $\IS$ with $\sum_{v \in N(\IS)} d(v) = \Omega(|E(G)|)$ in $O(1)$ \MPC rounds.

As one can see, Algorithm \ref{alg:MIS} is very similar to Algorithm \ref{alg:MM}, and the major difference is that in Algorithm \ref{alg:MIS} we sub-sample nodes instead of edges, since we cannot afford to have removed any edges between nodes we are considering for our independent set $\IS$.

Similarly to matching (cf. Section \ref{subsec:comput-i-B-E}), computing $i$, $\B$ and $\J_0$ can be completed in $O(1)$ \MPC rounds using several applications of Lemma \ref{lem:comm}. Therefore in the following Sections \ref{subsec:C^*}--\ref{subsec:->IS} we will first show how to deterministically construct in $O(1)$ \MPC rounds an appropriated set $\J' \subseteq \J_0$ that induces a low degree subgraph and then how to deterministically find in $O(1)$ \MPC rounds an independent set $\IS \subseteq \J'$ such that $\sum_{v \in N(\IS)} d(v) = \Omega(|E(G)|)$.

\subsection{Deterministically selecting $\J' \subseteq \J_0$ that induces a low degree subgraph}
\label{subsec:C^*}

We will show now how to deterministically, in $O(1)$ stages, find a subset $\J'$ of $\J_0$ that induces a low degree subgraph, as required in our \MPC algorithm for MIS. For that, our main goal is to ensure that every node has degree $O(n^{4\delta})$ in $\J'$ (to guarantee that its 2-hop neighborhood fits a single \MPC machine with $\spac = O(n^{8\delta})$), and that one can then locally find an independent $\IS \subseteq \J'$ that covers a linear number of edges.

We again proceed in $i-4$ stages (if $i \le 4$, then similarly to Algorithm \ref{alg:MM}, we will use $\J' = \J_0$), starting with $\J_0$ and sampling a new set $\J_j$ ($\J_j \subseteq \J_{j-1}$) in each stage $j = 1, 2, \dots, i-4$. The \emph{invariant} we will maintain is that, after every stage $j$, $0 \le j \le i-4$:

\begin{enumerate}[(i)]
\item all nodes $v \in \J_j$ have $d_{\J_j}(v) \le (1+o(1)) n^{-j\delta} d(v)$, and
\item all nodes $v \in \B$ have $\sum_{u \in \J_j \sim v} \frac{1}{d(u)} \ge \frac{\delta-o(1)}{3n^{\delta j}}$.
\end{enumerate}

It is easy to see that the invariant holds for $j = 0$ trivially, by definition of $\J_0$ and $\B$. In what follows, we will show how, for a given set $\J_{j-1}$ satisfying the invariant, to construct in $O(1)$ rounds a new set $\J_j \subseteq \J_{j-1}$ that satisfies the invariant too.

\paragraph{Distributing edges and nodes among the machines.}

In order to implement our scheme in the \MPC model, we first allocate the nodes and the edges of the graph among the machines.
\begin{itemize}
\item Each node $v$ in $\J_{j-1}$ distributes its adjacent edges to nodes in $\J_{j-1}$ across a group of machines (\emph{type $\J$ machines}), with at most one machine having fewer than $n^{4\delta}$ edges and all other machines having exactly $n^{4\delta}$ edges.%\Artur{Is this clear enough that all but one machine have exactly $n^{4\delta}$ edges? The statement here doesn't give any upper bound for the number of edges allocated to the machines \dots Maybe one should rephrase it.}
\item Each node $v$ in $\B$ distributes its adjacent edges to nodes in $\J_{j-1}$ across a group of machines (\emph{type $\B$ machines}), with at most one machine having fewer than $n^{4\delta}$ edges and all other machines having exactly $n^{4\delta}$ edges.
\end{itemize}

Note that nodes may be in both $\J_{j-1}$ and $\B$ and need only one group of machines, but for the ease of analysis we treat the groups of machines separately. Similarly to Section \ref{subsec:E^*}, type $\J$ machines will ensure Invariant (i) and type $\B$ machines will ensure Invariant~(ii).

In order to select $\J_{j-1} \subseteq \J_j$, we will first fix a seed specifying a hash function from a $\mathcal{H}$. Each hash function $h$ induces a candidate set $\J_h$ into which node in $\J_{j-1}$ is ``sampled with probability $n^{-\delta}$'', by placing $v$ into $\J_h$ iff $h(v)\le n^{3-\delta}$.

\paragraph{Type $\J$ machines.}

Consider a type $\J$ machine $x$ that gets allocated edges $V(x) \subseteq \J_{j-1}$ and let $v_x := |V(x)|$. For hash function $h \in \mathcal{H}$, we say $x$ is \emph{good} if $|V(x) \cap \J_h| \le v_x n^{-\delta} + n^{0.1\delta} \sqrt{v_x}$.

Each of the indicator random variables $\mathbf{1}_{\{v\in \J_h\}}$ is $\cj$-wise independent, and has expectation $n^{-\delta}$. Therefore we can apply to these random variable Lemma \ref{lem:conc}: taking $Z$ to be the sum of the indicator variables for $V(x)$ (i.e., $Z = |V(x)\cap \J_h|$), and choosing a sufficiently large constant $\cj$, Lemma \ref{lem:conc} implies that $\Prob{|Z - \mu| \ge n^{0.1\delta} \sqrt{v_x }} \le n^{-5}$. This means that with high probability, $|V(x) \cap \J_h| \le v_x n^{-\delta} + n^{0.1\delta} \sqrt{v_x}$, and $x$ is good.

\paragraph{Type $\B$ machines.}

Consider a type $\B$ machine $x$ that gets allocated edges $V(x) \subseteq \J_{j-1}$; let $v_x := |V(x)|$. For $h \in \mathcal{H}$, we call $x$ \emph{good} if $\sum_{v \in V(x) \cap \J_h} \frac{1}{d(v)} \ge n^{-\delta} \sum_{v \in V(x)} \frac{1}{d(v)} - n^{(0.9-i)\delta} \sqrt{v_x }$.

As before, we will apply Lemma \ref{lem:conc}, setting $Z_v = \frac{n^{(i-1)\delta}}{d(v)} \mathbf{1}_{\{v\in \J_h\}}$ and $Z = \sum_{v \in V(x)} Z_v$. Since $V(x) \subseteq \J$, each $d(v)$ is at least $n^{(i-1)\delta}$, and so the variables $Z_v$ take values in $[0,1]$. They have expectation $\Exp{Z_v} = \frac{n^{(i-2)\delta}}{d(v)}$, and as before, they are $\cj$-wise independent. Hence, we can apply Lemma \ref{lem:conc} with sufficiently large $c$ to find that $\Prob{|Z - \mu| \ge n^{0.1\delta} \sqrt{v_x }} \le n^{-5}$. Hence, with high probability,
\begin{align*}
n^{(i-1)\delta} \sum_{v \in V(x) \cap \J_h} \frac{1}{d(v)} &\ge
n^{(i-2)\delta} \sum_{v \in V(x)} \frac{1}{d(v)} - n^{0.1\delta} \sqrt{v_x}
\enspace,
\end{align*}
and therefore $\sum_{v \in V(x) \cap \J_h}\frac{1}{d(v)} \ge n^{-\delta} \sum_{v \in V(x)} \frac{1}{d(v)} - n^{(0.9-i)\delta} \sqrt{v_x}$, so $x$ is good.

Since there are at most $\frac{2n^2}{\spac} + 2 n \le n^2$ machines, by a union bound the probability that a particular hash function $h \in \mathcal{H}$ makes all machines good is at least $1 - n^{-3}$. The expected number of machines which are not good for a random choice of function is therefore less than $1$. So, by the method of conditional expectations (cf. Section \ref{sec:condexp}), using objective $q_x(h) = \textbf{1}_{\text{$x$ is good for $h$}}$, in a constant number of \MPC rounds we can find a hash function $h \in \mathcal{H}$ which makes all machines good. We then use such hash function $h$ to set $\J_j = \J_h$.

\subsubsection{Properties of $\J_j$: satisfying the invariant}
\label{subsub:properties-of-Jj-satisfying-invariant}

Having fixed a sub-sampled set of nodes $\J_j$ for the stage, we need to show that since all machines were good, we satisfy our invariant for the stage. \junk{Missing proofs are deferred to Appendix \ref{app:mis}.}

\begin{lemma}[Invariant (i)]
\label{lemma-mis-inv-i}
All nodes $v \in \J_j$ satisfy
\begin{align*}
    d_{\J_j}(v) &\le (1+o(1)) n^{-j\delta} d(v)\enspace.
\end{align*}
\end{lemma}
\junk{
\begin{lemma}
\label{lemma-mis-inv-i}
All nodes $v \in \J_j$ satisfy $d_{\J_j}(v) \le (1+o(1)) n^{-j\delta} d(v)$.
\end{lemma}
}

\def\APPENDMISINVONE{
\begin{proof}%[Proof of Lemma \ref{lemma-mis-inv-i}]
Node $v$'s neighbors in $\J_{j-1}$ were divided among $\lfloor\frac{d_{\J_{j-1}}(v)}{n^{3\delta}}\rfloor$ type $\J$ machines containing $n^{4\delta}$ neighbors, and one type $\J$ machine containing the remaining $d_{\J_{j-1}}(v) - n^{4\delta} \lfloor\frac{d_{\J_{j-1}}(v)}{n^{4\delta}}\rfloor \le d_{\J_{j-1}}(v)$ neighbors. Therefore we obtain,
	\begin{align*}
	d_{\J_j}(v)
	&\le
	\sum_{\text{$v$'s machines }x}v_x n^{-\delta}+n^{0.1\delta}\sqrt{v_x }\\
	&=
	n^{-\delta} d_{\J_{j-1}}(v) +
	\left\lfloor\frac{d_{\J_{j-1}}(v)}{n^{4\delta}}\right\rfloor
	n^{0.1\delta}\sqrt{n^{4\delta}} +
	n^{0.1\delta} \sqrt{d_{\J_{j-1}}(v)}
	\\
	&\le
	n^{-\delta} d_{\J_{j-1}}(v) +
	\frac{n^{0.1\delta}}{n^{2\delta}}d_{\J_{j-1}}(v) +
	n^{0.1\delta} \sqrt{d_{\J_{j-1}}(v)}
	\enspace.
	\end{align*}
	
	If $d_{\J_{j-1}}(v) \ge n^{3\delta}$, we have
	\begin{align*}
	d_{\J_j}(v)
	&\le
	n^{-\delta} d_{\J_{j-1}}(v) +
	\frac{n^{0.1\delta}}{\sqrt s} d_{\J_{j-1}}(v) +
	n^{0.1\delta}\sqrt{d_{\J_{j-1}}(v)}\\
	&\le
	n^{-\delta}d_{\J_{j-1}}(v)+
	n^{-1.9\delta}d_{\J_{j-1}}(v) +
	n^{-1.4\delta}d_{\J_{j-1}}(v)
	\\
	&=
	(1+o(1))n^{-\delta}d_{\J_{j-1}}(v)
	\le
	(1+o(1))n^{-j\delta}d(v)
	\enspace.
	\end{align*}
	
Otherwise, $d_{\J_{j}}(v) \le d_{\J_{j-1}}(v) \le n^{3\delta} \le n^{-j\delta}d(v)$. In either case, we satisfy the invariant for stage $j$.
\end{proof}
}
\APPENDMISINVONE

\begin{lemma}[Invariant (ii)]
\label{lemma-mis-inv-ii}
All nodes $v \in \B$ satisfy
\begin{align*}
    \sum_{u \in \J_j \sim v} \frac{1}{d(u)} &\ge \frac{\delta-o(1)}{4n^{\delta j}} \enspace.
\end{align*}
\end{lemma}
\junk{
\begin{lemma}
\label{lemma-mis-inv-ii}
All nodes $v \in \B$ satisfy $\sum_{u \in \J_j \sim v} \frac{1}{d(u)} \ge \frac{\delta-o(1)}{4n^{\delta j}}$.
\end{lemma}
}

\def\APPENDMISINVTWO{
\begin{proof}%[Proof of Lemma \ref{lemma-mis-inv-ii}]
	Node $v$'s neighbors in $\J_{j-1}$ were again divided among $\lfloor\frac{d_{\J_{j-1}}(v)}{n^{4\delta}}\rfloor$ type $\B$ machines containing $n^{4\delta}$ neighbors, and one type $\B$ machine containing the remaining $d_{\J_{j-1}}(v)-n^{4\delta}\lfloor\frac{d_{\J_{j-1}}(v)}{n^{4\delta}}\rfloor\le d_{\J_{j-1}}(v)$ neighbors. Denote $y:= \sum_{u\in \J_{j-1}\sim v}\frac{1}{d(u )}$ for brevity.
	\begin{align*}
	\sum_{u\in \J_j\sim v}\frac{1}{d(u)}
	&=
	\sum_{\text{$v$'s machines $x$ }}\sum_{u\in V(x)\cap \J_j}\frac{1}{d(u)}\\
	&\ge
	\sum_{\text{$v$'s machines $x$ }}
	\left(
	n^{-\delta}\sum_{u \in V(x)}\frac{1}{d(u)} -
	n^{(1.9-i) \delta}\sqrt{v_x}
	\right)
	\\
	&\ge
	n^{-\delta}
	\left(
	\sum_{u\in \J_{j-1}\sim v}\frac{1}{d(u)} -
	n^{(1.9-i)\delta}
	\left(
	\left\lfloor\frac{d_{\J_{j-1}}(v)}{n^{4\delta}}\right\rfloor
	\sqrt n^{4\delta} +
	\sqrt{d_{\J_{j-1}}(v)}
	\right)
	\right)\\
	&\ge
	n^{-\delta}
	\left(
	y -
	n^{(1.9-i)\delta}
	\left(
	n^{-2\delta} d_{\J_{j-1}}(v) + \sqrt{d_{\J_{j-1}}(v)}
	\right)
	\right)
	\enspace.
	\end{align*}
	
We know that $d_{\J_{j-1}}(v) \le n^{i\delta} \sum_{u \in \J_{j-1} \sim v} \frac{1}{d(u)} = n^{i\delta} y$, since all nodes $u \in \J_{j-1}$ are in $\J$ and have degree at most $n^{i\delta}$. So,
	\begin{align*}
	\sum_{u \in \J_j \sim v} \frac{1}{d(u)}
	&\ge
	n^{-\delta}
	\left(
	y -
	n^{(1.9-i)\delta}
	\left(
	n^{-2\delta}d_{\J_{j-1}}(v) + \sqrt{d_{\J_{j-1}}(v)}
	\right)
	\right)\\
	&\ge
	n^{-\delta}
	\left(
	y -
	n^{(1.9-i)\delta}
	\left(
	n^{-2\delta}n^{i\delta} y + \sqrt{n^{i\delta} y}
	\right)
	\right)
	\\
	&=
	n^{-\delta}\left(y-n^{-0.1\delta}y - n^{(1.9-0.5i)\delta}\sqrt y\right)\\
	&=
	n^{-\delta}\left(y-o(y) - n^{(1.9-0.5i)\delta}\sqrt y\right)
	\enspace.
	\end{align*}
	
From our invariant we know that $y \ge \frac{\delta-o(1)}{4n^{\delta(j-1)}}$, and so $\sqrt{y} \le \frac{y}{\sqrt\frac{\delta-o(1)}{4n^{\delta(j-1)}}} \le
	\frac{y}{\delta}n^{0.5\delta(j-1)}$. Hence,
	\begin{align*}
	\sum_{u \in \J_j\sim v} \frac{1}{d(u)}
	&\ge
	n^{-\delta}\left(y-o(y) - n^{(1.9-0.5i)\delta}\sqrt y\right)\\
	&\ge
	n^{-\delta}
	\left(
	y-o(y) - n^{(1.9-0.5i)\delta}\cdot \frac{y}{\delta}n^{0.5\delta(j-1)}
	\right)
	\\
	&\ge
	n^{-\delta}
	\left(y-o(y) - \frac{y}{\delta} \cdot n^{(1.4-0.5i+0.5j)\delta} \right)\\
	&\ge
	n^{-\delta}
	\left(y-o(y) - \frac{y}{\delta}\cdot n^{(1.4-0.5i+0.5(i-4))\delta} \right)
	\\
	&\ge
	n^{-\delta}\left(y-o(y) -\frac{y}{\delta}\cdot n^{-0.6\delta} \right)
	\ge
	n^{-\delta}\left(y-o(y)\right)\\
	&\ge
	n^{-\delta} \cdot \frac{\delta-o(1)}{3n^{\delta(j-1)}}
	\ge
	\frac{\delta-o(1)}{3n^{\delta j}}
	\enspace.
	\qedhere
	\end{align*}
\end{proof}
}
\APPENDMISINVTWO

Since our invariant is preserved in every stage, it holds in our final sub-sampled node set $\J' := \J_{i-4}$.

\subsection{Finding an independent set \IS}
\label{subsec:->IS}

After $i-4$ stages, we now have a node set $\J' := \J_{i-4}$ with the following properties (cf. Lemmas \ref{lemma-mis-inv-i} and \ref{lemma-mis-inv-ii}):

\begin{enumerate}[(i)]
\item all nodes $v \in \J'$ have $d_{\J'}(v) \le (1+o(1)) n^{(4-i)\delta} d(v) \le 2n^{4\delta}$;
\item all nodes $v \in \B$ have $\sum_{u \in \J' \sim v} \frac{1}{d(u)} \ge \frac{\delta-o(1)}{3n^{(i-4)\delta}}$.
\end{enumerate}
(If $i \le 4$, we instead have that for $v\in \J'$, $d_{\J'}(v) \le n^{i\delta}$ and for $v \in B$, $\sum_{u \in \J' \sim v} \frac{1}{d(u)} \ge \frac{\delta}{3}$ from setting $\J'=\J_0$).

We now show a property analogous to Lemma \ref{lem:egoodsample} (and hence Lemma \ref{lem:Luby}) for the node set $\J'$\junk{(see Appendix \ref{app:goodsample} for a proof)}.

\begin{lemma}
\label{lem:goodsample}
For each node $v \in \B$, either $v$ has a neighbor $u \in \J'$ with $d_{\J'}(u) = 0$ or $v$ satisfies $\sum_{u \in \J' \sim v} \frac{1}{d_{\J'}(u )} \ge 0.1 \delta$.
\end{lemma}

\def\APPENDANOTHERONE{
\begin{proof}%[Proof of Lemma \ref{lem:goodsample}]
The claim trivially holds $i\le 4$. Otherwise, for every $v \in \B$ with no neighbor $u \in \J'$ with $d_{\J'}(u) = 0$, we have the following,
\begin{align*}
	\sum_{u \in \J' \sim v} \frac{1}{d_{\J'}(u)}
	&\ge
	\sum_{u \in \J' \sim v}\frac{1}{2n^{(4-i)\delta}d(u)}
	=
	\frac{1}{2n^{(4-i)\delta}} \sum_{u \in \J' \sim v}\frac{1}{d(u)}
	\ge
	\frac{1}{2n^{(4-i)\delta}} \cdot \frac {\delta}{5n^{(i-4)\delta}}
	=
	0.1\delta
	\enspace,
\end{align*}
where the first inequality follows from Invariant (i) (Lemma \ref{lemma-mis-inv-i}) and the second one from Invariant (ii) (Lemma \ref{lemma-mis-inv-ii}).
\end{proof}
}
\APPENDANOTHERONE

Now we are ready to present our deterministic \MPC algorithm that for a given subset of nodes $\J'$ satisfying the invariant, in $O(1)$ rounds constructs an independent set $\IS \subseteq \J'$ such that the removal of $\IS \cup N(\IS)$ removes $\Omega(\delta |E|)$ edges from the graph.

Each node $v \in \B$ is assigned a machine $x_v$ which gathers a set $N_v$ of up to $n^{4\delta}$ of $v$'s neighbors in $\J'$ (if $v$ has more than $n^{4\delta}$ neighbors in $\J'$, then take an arbitrary subset of $n^{4\delta}$ of them), along with all of their neighborhoods in $\J'$ (i.e., $N_{\J'}(N_v)$). By Invariant (i), this requires at most $n^{4\delta} \cdot 2n^{4\delta} = O(n^{8\delta})$ space per machine. Since $|\B| \le n$, this is $O(n^{1+8\delta})$ total space.
We prove \junk{(cf. Appendix~\ref{app:goodsample2}) }that these sets $N_v$ preserve the desired property:

\begin{lemma}
\label{lem:goodsample2}
Each node $v \in \B$ either has a neighbor $u \in \J'$ with $d_{\J'}(u) = 0$ or satisfies $\sum_{u \in N_v} \frac{1}{d_{\J'}(u)} \ge 0.1 \delta$.
\end{lemma}

\def\APPENDMISSAM{
\begin{proof}%[Proof of Lemma \ref{lem:goodsample2}]
If $d_{\J'}(v) \le n^{4\delta}$ then $N_v = N_{\J'}(v)$ and so the lemma holds by Lemma \ref{lem:goodsample}.
	
Otherwise we have $|N_v| = n^{4\delta}$, and so by Invariant (i) in the first inequality, we get,
\begin{align*}
	\sum_{u \in N_v}\frac{1}{d(u )}
        &\ge
	\sum_{u \in N_v}\frac{1}{2n^{4\delta}}
        =
	\frac{n^{4\delta}}{2n^{4\delta}}
        =
	\frac12
        >
	0.1\delta
	\enspace.
	\qedhere
\end{align*}
\end{proof}
}
\APPENDMISSAM

We now do one further derandomization step to find an independent set. We will fix a seed specifying a hash function from $\mathcal{H}$. This hash function $h$ will be used to map each node $v$ in $\J'$ to a value $z_v \in [n^{3}]$. Then, $v$ joins the \emph{candidate independent set} $\IS_h$ iff $z_v < z_u$ for all $u \sim v$.

Clearly $\IS_h$ is indeed an independent set, but we want to show that removing $\IS_h \cup N(\IS_h)$ from the graph reduces the number of edges by a constant fraction. We will show that in expectation (over a random choice of $h \in \mathcal{H}$) this is indeed the case, and then we can apply the method of conditional expectations (cf. Section \ref{sec:condexp})to conclude the construction.

Each machine $x_v$ is \emph{good} for a hash function $h \in \mathcal{H}$ if it holds a node $u \in N_v \cap \IS_h$. Since $x_v$ holds the neighborhoods in $\J'$ of nodes in $N_v$, it can determine whether they are members of $\IS_h$. We show that with constant probability, $x_v$ is good for a random hash function $h \in \mathcal{H}$\junk{ (see Appendix \ref{app:goodhash} for a proof)}.

\begin{lemma}
\label{lem:goodhash}
For any machine $x_v$ holding a set $N_v$ and its neighborhood in $\J'$, with probability at least $0.01 \delta$ (over a choice of a random hash function $h \in \mathcal{H}$) it holds that $|N_v \cap \IS_h| \ge 1$.
\end{lemma}

\def\APPENDMISHASH{
\begin{proof}%[Proof of Lemma \ref{lem:goodhash}]
If any node $u \in N_v$ has $d_{\J'}(u) = 0$, it will join $\IS_h$ and we are done.
	
Otherwise, by Lemma \ref{lem:goodsample2} we have $\sum_{u \in N_v} \frac{1}{d_{\J'}(u)} \ge 0.1 \delta$ and we will consider, for the sake of the analysis, some arbitrary subset $N^*_v \subseteq N_v$ such that $1 \ge \sum_{u \in N^*_v} \frac{1}{d_{\J'}(u)} \ge 0.1 \delta$.
	
For any $u \in N^*_v$ we have
\begin{align*}
	\frac{1}{3d_{\J'}(u)} - \frac1{n^3}
	&\le
	\Prob{z_u < \frac{n^3}{3d_{\J'}(u)}}
	\le
	\frac{1}{3d_{\J'}(u)}
	\enspace.
\end{align*}
	
Let $\mathcal{A}_u$ be the event $\left\{u \in \IS_h \land z_u < \frac{n^3}{3d_{\J'}(u)}\right\}$. By pairwise independence,
\begin{align*}
	\Prob{A_u}
	&=
	\Prob{u \in \IS_h \land z_u < \frac{n^3}{3d_{\J'}(u)}}
	\\
	&=
	\Prob{z_u < \frac{n^3}{3d_{\J'}(u)}} -
	\Prob{u \notin \IS_h \land z_u < \frac{n^3}{3d_{\J'}(u)}}
	\\
	&=
	\Prob{z_u < \frac{n^3}{3d_{\J'}(u)}} -
	\Prob{\bigcup_{w \in \J' \sim u}\left\{z_w\le z_u  < \frac{n^3}{3d_{\J'}(u)}\right\}}
	\\
	&\ge
	\frac{1}{3d_{\J'}(u)} - \frac1{n^3} -
	\sum_{w \in \J' \sim u} \Prob{z_w \le z_u  < \frac{n^3}{3d_{\J'}(u)}}
	\\
	&\ge
	\frac{1}{3d_{\J'}(u)} - \frac1{n^3} -
	\sum_{w \in \J' \sim u}
	\Prob{z_w, z_u  < \frac{n^3}{3d_{\J'}(u)}}
	\\
	&=
	\frac{1}{3d_{\J'}(u)} - \frac1{n^3} - d_{\J'}(u) \cdot \frac{1}{(3d_{\J'}(u))^2}
%	\\&
    =
	\frac{2}{9d_{\J'}(u)} - \frac1{n^3}
	\enspace.
\end{align*}
	
Furthermore, for any $u \ne u'$, again by pairwise independence of hash functions from~$\mathcal{H}$,
	
\begin{align*}
	\Prob{A_u \cap A_{u'}}
	&=
	\Prob{u \in \IS_h \land z_u < \frac{n^3}{3d_{\J'}(u)} \land u' \in \IS_h
		\land z_{u'} < \frac{n^3}{3d_{\J'}(u')}}
	\\
	&\le
	\Prob{z_u < \frac{n^3}{3d_{\J'}(u)} \land z_{u'} < \frac{n^3}{3d_{\J'}(u')}}
	\le
	\frac{1}{9 \cdot d_{\J'}(u) \cdot d_{\J'}(u')}
	\enspace.
\end{align*}
	
So, by the principle of inclusion-exclusion,
\begin{align*}
	\Prob{\bigcup_{u \in N^*_v} A_u}
	&\ge
	\sum_{u \in N^*_v} \Prob{A_u} - \sum_{\{u,u'\} \subseteq N^*_v, u \ne u'} \Prob{A_u \cap A_{u'}}
	\\
	&\ge
	\sum_{u \in N^*_v} \left(\frac{2}{9d_{\J'}(u)} - \frac1{n^3}\right) -
	\sum_{\{u,u'\} \subseteq N^*_v, u \ne u'} \frac{1}{9d_{\J'}(u)d_{\J'}(u')}
	\\
	&\ge
	\sum_{u \in N^*_v} \frac{2}{9d_{\J'}(u)}  -
	\frac{1}{18}\left(\sum_{u \in N^*_v} \frac{1}{d_{\J'}(u)}\right)^2 - \frac1{n^2}
	\\
	&\ge
	\frac29 \sum_{u \in N^*_v} \frac{1}{d_{\J'}(u)} -
	\frac{1}{18} \sum_{u \in N^*_v}\frac{1}{d_{\J'}(u)} - \frac1{n^2}
	\\
	&=
	\frac16 \sum_{u \in N^*_v} \frac{1}{d_{\J'}(u)} - \frac1{n^2}
	\\
	&\ge
	0.016 \delta - \frac1{n^2}
	\enspace.
\end{align*}
	
Here we use that, by our choice of $N^*_v$, we have $\left(\sum_{u \in N^*_v} \frac{1}{d_{\J'}(u)}\right)^2 \le \sum_{u \in N^*_v} \frac{1}{d_{\J'}(u)} \le 1$ and $0.1 \delta \le \sum_{u \in N^*_v} \frac{1}{d_{\J'}(u)}$.
%, and also apply Lemma \ref{lem:goodsample2}.\UArtur{Where do you use Lemma \ref{lem:goodsample2}?}
	
Therefore, for $n$ larger than a sufficiently large constant, we have the following,
\begin{align*}
	\Prob{N_v \cap \IS_h \ne \emptyset} &
	%        =
	%    \Prob{\bigcup_{u \in N^*_v} \{u \in \IS_h\}}
	\ge
	\Prob{\bigcup_{u \in N^*_v} A_u}
	\ge
	0.016 \delta - \frac1{n^2}
	\ge
	0.01 \delta
	\enspace.
	\qedhere
\end{align*}
\end{proof}
}
\APPENDMISHASH

For a hash function $h \in \mathcal{H}$, we will denote $N_h := \{v \in \B: N_v \cap \IS_h \ne \emptyset\}$, i.e., the set of nodes to be removed if the independent set induced by hash function $h$ is chosen. By Lemma~\ref{lem:goodhash} and by the definition of $\B$ (which ensures $\sum_{v \in \B} d(v) \ge \frac\delta 2|E|$),
\begin{align*}
    \Exp{\sum_{v \in N_h}d(v)}
        &\ge
    \sum_{v \in \B} d(v) \cdot \Prob{v \in N_h}
        \ge
    0.01  \delta \sum_{v \in \B} d(v)
        \ge
    \frac{\delta^2 |E|}{200}
        \enspace.
\end{align*}

By the method of conditional expectations (cf. Section \ref{sec:condexp}), using $q_{x_v}(h) = d(v) \textbf{1}_{\text{$x_v$ is good for $h$}}$, we can select a hash function $h$ with $\sum_{v \in N_h} d(v) \ge \frac{ \delta^2}{200} \cdot |E|$. We then add the independent set $\IS := \IS_h$ to our output, and remove $\IS$ and $N(\IS)$ from the graph. In doing so, we remove at least $\frac12 \sum_{v \in N_h} d(v) \ge \frac{\delta^2|E|}{400}$ edges from the graph.

\subsection{Completing the proof of Theorem \ref{thm:MIS}: finding MIS}

Now we are ready to complete the proof of Theorem \ref{thm:MIS}. Our algorithm returns an MIS in at most $\log_\frac{1}{1-\delta^2/200}|E| = O(\log n)$ stages, each stage of the algorithm, as described above, taking a constant number of rounds in \MPC. The space required is dominated by storing the input graph $G$ ($O(m)$ total space) and collecting node neighborhoods when finding an independent set ($O(n^{8\delta})$ space per machine, $O(n^{1+8\delta})$ total space). Setting $\delta = \frac{\eps}{8}$ allows us to conclude Theorem \ref{thm:MIS} by obtaining that for any constant $\eps > 0$, MIS can be found deterministically in \MPC in $O(\log n)$ rounds, using $O(n^{\eps})$ space per machine and $O(m+n^{1+\eps})$ total space.
\qed

%---------------------------------------------------------------------------------------------------------------------------------------------------------

%---------------------------------------------------------------------------------------------------------------------------------------------------------

\section{MIS and maximal matching in $O(\log\Delta+\log\log n)$ \MPC rounds}
\label{sec:logD-rounds}
%---------------------------------------------------------------------------------------------------------------------------------------------------------

While our main efforts in Sections \ref{sec:max-matching}--\ref{sec:MIS} was on achieving deterministic MIS and maximal matching algorithms running in $O(\log n)$ \MPC rounds, with some additional work we can improve them for graphs with low maximum degree $\Delta$, obtaining deterministic $O(\log\Delta + \log\log n)$-round \MPC algorithms, where $\Delta$ is the maximum degree in the input graph. In the following, we will present our algorithms for MIS. Obtaining similar bounds for maximal matching %and (deg+1)-list coloring
will be done by reducing to the MIS problem. Indeed, these well-known reductions dating back to Luby \cite{Luby86} can be efficiently implemented in the low-space \MPC setting, provided that $\Delta$ is sufficiently small: $O(n^{\epsilon/c})$ for small constant~$c$. %We will elaborate more about implementing these reductions towards the end of this section.
%
%
%noting, as for Corollary \ref{cor:congcMM}, that for maximal matching one can use the standard reduction of performing MIS on the line graph, when $\log \Delta = o(\log n)$

Let us first observe that it is sufficient to consider the case where $\Delta \le n^{\delta}$, as otherwise we can use the $O(\log n)$ algorithm from Theorem \ref{thm:MIS} to achieve an $O(\log\Delta)$-round \MPC algorithm.

Assuming that $\Delta \le n^{\delta}$ (where $\delta$ is, as before, a constant sufficiently smaller than \eps) we mimic Luby's algorithm for MIS (cf. Algorithm~\ref{alg:MIS} in Section \ref{sec:MIS}). We group the rounds of the algorithm into \emph{stages}, each stage consisting of $\ell = O(\delta \log_{\Delta} n)$ \emph{phases}. As in Algorithm~\ref{alg:MIS}, in each phase $i$, we find an independent set $\IS_i$ in the current graph $G_{i-1}$ (after $i-1$ phases) and remove $\IS_i \cup N(\IS_i)$ to obtain a new graph $G_i$. This process terminates when $G_i$ is the empty graph, in which case the union of all independent sets found is an MIS. We will show that after an $O(\log\log n)$ rounds preprocessing, each stage (consisting of $\ell$ phases) can be implemented deterministically in a constant number of rounds in the \MPC model, using $O(n^{\eps})$ space per machine and $O(n^{1+\eps})$ total space. Then we will show that the algorithm terminates in $O(\log n)$ phases, and hence in $O(\log\Delta)$ stages.

We will first design a randomized algorithm and then show how the features of our algorithm lead to a deterministic algorithm, following the approach presented in earlier sections.

%---------------------------------------------------------------------------------------------------------------------------------------------------------

\subsection{Randomized \MPC algorithm with smaller seed}
\label{sec:logD-rounds-random}

We begin with a presentation of a simple randomized MIS algorithm that uses random seeds of $O(\log \Delta)$ bits for each step in Luby's algorithm. We later show how to adjust this randomized algorithm to be efficiently implemented deterministically. Let us first observe that in Luby's algorithm we only need independence between nodes that are $2$-hop apart. This allows us to reduce the seed length from $O(\log n)$ to $O(\log \Delta)$ by assigning every node a new name with only $O(\log \Delta)$ bits, such that every $2$-hop neighbors are given distinct names. This task can be stated as a vertex coloring in the graph $G^2$. For every graph $G$ with maximum degree $\Delta$, Linial \cite{Linial92} showed an $O(\Delta^2)$-coloring using $O(\log^*n)$ rounds in the \local model, and Kuhn \cite{Kuhn09} extended this algorithm and show that it can implemented also in the \congest model within the same number of rounds.

In our context, since we wish to color the graph $G^2$, we need compute a $O(\Delta^4)$-coloring $\chi$ on $G^2$. As $\Delta \le n^{\delta}$ (for some constant $\delta$), the $2$-hop neighbors of each node can be stored on a single machine, and thus we can efficiently simulate the coloring algorithm of \cite{Kuhn09} within $O(\log^* n)$ rounds. Once the $O(\Delta^4)$-coloring is computed, the algorithm will simulate the random choices of Luby's algorithm using a family $\mathcal{H}^*$ of pairwise independent hash functions (as in Lemma~\ref{lem:hash}) mapping the $O([\Delta^4])$ colors to number in $[\Delta^4]$. The benefit of using this family of hash functions $\mathcal{H}^*$ is that one can pick a function $h \in \mathcal{H}^*$ i.u.r. using only $O(\log \Delta)$ bits (rather than $O(\log n)$ bits).
%
%\paragraph{$O(\Delta^4)$-coloring $\chi$ of $G^2$.}
%%
%To ensure that each stage can be implemented deterministically in a constant number of rounds on \MPC, we will deterministically find a coloring $\chi$ of graph $G^2$ with $O(\Delta^4)$ colors. For that, we use the classical algorithm due to Linial \cite{Linial92} that performs this task in $O(\log^*\Delta)$ rounds. (Linial's algorithm takes an input graph of maximum degree $\bar\Delta'$ and finds its $O(\bar\Delta^2)$ coloring in $O(\log^*\bar\Delta)$ rounds in the \local model, and it can be easily implemented in the \MPC model in $O(\log^*\Delta)$ rounds with at least $O(\bar\Delta^2)$ space per machine. But since we consider $G^2$, the maximum degree is $\Delta^2$ and the number of colors becomes $O(\Delta^4)$.)\Artur{Maybe one could elaborate more about the use of the coloring.}
%
%\paragraph{Hash functions.}
%%
%Since in our randomized algorithm (to be derandomized later), in each round we need independence in the $2$-hop neighborhood, we will use hash functions that make pairwise-independent decisions inside each $2$-hop neighborhood. Thus, we will use a family $\mathcal{H}^*$ of pairwise independent hash functions (as in Lemma~\ref{lem:hash}) mapping the $O([\Delta^4])$ colors to $[\Delta^4]$.

\paragraph{The randomized MIS algorithm.}
We will find an MIS stage by stage. Fix a stage and consider its $\ell = O(\delta \log_{\Delta}n)$ phases. These phases are identical in effect to those of Algorithm \ref{alg:MIS}, except that the hash functions map color classes rather than individual nodes. Since behavior is independent of nodes outside of 2-hop neighborhoods, and we have ensured that no nodes within distance two are colored the same, we can show via the same method that we remove a constant fraction of $G_{i-1}$'s edges each phase. Therefore, after $O(\log n)$ phases we have removed all edges from the graph and thus have found an MIS. So, we require only $O\left(\frac{\log n}{\delta\log_{\Delta}n}\right) = O(\log\Delta)$ stages.
%

%---------------------------------------------------------------------------------------------------------------------------------------------------------

\subsection{Deterministic MIS algorithm in $O(\log\Delta + \log\log n)$ \MPC rounds}
\label{sec:logD-rounds-deterministic}

In this section, we show how to derandomize the algorithm from Section \ref{sec:logD-rounds-random} in only $O(1)$ rounds per stage, with $O(\log\log n)$ total overhead; that is, in $O(\log\Delta + \log\log n)$ \MPC rounds overall, completing the proof of Theorem~\ref{thm:MM-MIS-Delta}.

\subsubsection{Basic tools}

\paragraph{Gathering the $r$-th hop neighborhood.}
Let $r = O(\delta\log_{\Delta}n)$. Let $G$ be the input graph. Since the $r$-th hop neighborhood of any single node has at most $\Delta^r = O(n^{\delta})$ nodes, in $O(\log r)$ rounds we can collect the $r$-th hop neighborhood of each node $u$ in $G$ into the machine that stores $u$, ensuring that the space used on any single machine is $\spac = O(n^{\eps})$. With our choice of $r$, this operation can be implemented in $O(\log\log n)$ \MPC rounds.

\paragraph{Hash functions.}
We are using a family $\mathcal{H}^*$ of pairwise independent hash functions mapping $O([\Delta^4])$ to $[\Delta^4]$, as in Lemma \ref{lem:hash}. Each hash function from $\mathcal{H}^*$ requires $O(\log\Delta)$ bits to specify. Therefore, in a single stage, if we are using a sequence of hash functions $h_1, \dots, h_{\ell}$, then the total number of seeds used is $\ell \cdot O(\log\Delta) = O(\delta \cdot \log_{\Delta}n \cdot \log\Delta) = O(\delta \cdot \log n)$. Hence, all possible sequences $h_1, \dots, h_{\ell}$ of $\ell$ hash functions from $\mathcal{H}^*$ can be stored and evaluated on any single machine (with space $\spac = O(n^{\eps})$).

\subsubsection{Deterministic \MPC implementation}

Now, we are ready to present a deterministic \MPC algorithm that in $O(\log\Delta)$ stages finds MIS, and (after an $O(\log\log n)$-rounds preprocessing) requires only a constant number of rounds per stage. The algorithm is a derandomization of the algorithm presented in Section \ref{sec:logD-rounds-random}.

We will consider only the case $\Delta \le n^{\delta}$, and we assume that there are $O(n)$ machines available, each machine with space $\spac = O(n^{\eps})$; so the total space is $O(n^{1+\eps})$.

\paragraph{Preprocessing.}
In order to implement our scheme in the \MPC model, we first allocate all nodes and edges among the machines. We allocate each node to a separate machine; since $\spac = O(n^\eps)$ and $\Delta \le n^{\delta}$, we can distribute the edges so that each node $v$ has all its adjacent edges on the machine designated to $v$.

Next, we run Linial's algorithm to find a $O(\Delta^4)$-coloring of $G^2$. As discussed earlier, this can be done in $O(\log^*n)$ rounds in \MPC.

Next, for every node $v$, we distribute to the machine designated to $v$ the entire family $\mathcal{H}^*$ of pairwise independent hash functions mapping $O([\Delta^4])$ to $[\Delta^4]$, as discussed above.

Next, for every node $v$, we collect its $r$-th hop neighborhood into the machine that stores $u$, ensuring that the space used on any single machine is $\spac = O(n^{\eps})$. With our choice of $r = O(\delta\log_{\Delta}n)$, this operation can be implemented in $O(\log\log n)$ \MPC rounds.

\paragraph{Implementing a stage in $O(1)$ \MPC rounds.}
Consider a stage of the algorithm consisting of $\ell = O(\delta\log_{\Delta}n)$ phases; we set $r = 2 \ell$. Let $G_0$ be the graph at the beginning of the stage. Let $G_i$ denote the graph obtained after the $i$-th phase, $0 \le i \le \ell$; we will construct $G_1, \dots, G_{\ell}$ iteratively. We assume that each node $u$ knows its $r$-th hop neighborhood in $G_0$, stored on the machine associated with~$u$.

The randomized algorithm takes a sequence of hash functions $h_1, \dots, h_{\ell}$ i.u.r. from $\mathcal{H}^*$, and in each phase $1 \le i \le \ell$, for every node $v$ in $G_{i-1}$, it adds $v$ to $\IS_i$ iff $h_i(\chi(v)) < h_i(\chi(u))$ for all $u \sim_{G_{i-1}} v$; the new graph $G_i$ is obtained from $G_{i-1}$ by removal of $\IS_i \cup N(\IS_i)$.

To derandomize this algorithm, each node considers all possible sequences of hash functions $h_1, \dots, h_{\ell}$ from $\mathcal{H}^*$. The number of such sequences is $|\mathcal{H}^*|^\ell$, and hence they can be represented by $\ell \cdot O(\log\Delta) = O(\delta \log n)$ seeds, and be stored and evaluated on a single machine.

We go through all sequences of hash functions $\mathfrak{h} = \langle h_1, \dots, h_{\ell} \rangle$ from $\mathcal{H}^*$ and find independent sets $\IS_1, \dots, \IS_{\ell}$ for each such sequence; let $\IS^{\langle \mathfrak{h} \rangle}$ be the union of the independent sets found in a given stage for the sequence of hash functions $\mathfrak{h}$. Let $G^{\langle \mathfrak{h} \rangle}$ be the graph $G_{\ell}$ at the end of that stage obtained for the sequence of hash functions $\mathfrak{h}$.

At least one sequence $\mathfrak{h} = \langle h_1, \dots, h_{\ell} \rangle$ ensures that the resulting graph $G^{\langle \mathfrak{h} \rangle}$ has at most $c_{phase}^\ell \cdot |E_0|$ edges, where $c_{phase}$ is the constant fraction that we can guarantee we remove each phase. Our algorithm will select the sequence $\mathfrak{h}_0$ that minimizes the number of edges of the resulting graph $G^{\langle \mathfrak{h}_0 \rangle}$, and will output to the next stage the resulting graph $G^{\langle \mathfrak{h}_0 \rangle}$ and the obtained independent set~$\IS^{\langle \mathfrak{h}_0 \rangle}$.

Our central observation is that for a node $v$ in $G_0$, for a fixed sequence $h_1, \dots, h_{\ell}$ of hash functions from $\mathcal{H}^*$, one can decide whether the algorithm will select $v$ to be in $\IS_j$ or in $N(\IS_j)$ for some $1 \le j \le \ell$, solely on the basis of the $(2\ell)$-th hop neighborhood of $v$ in $G_0$. Indeed, to detect whether $v \in \IS_1 \cup N(\IS_1)$, node $v$ must only check whether $h_1(\chi(v)) < h_1(\chi(u))$ for all $u \sim_{G_0} v$, in which case $v \in \IS_1$, and whether for any $u \sim_{G_0} v$ it holds that $h_1(\chi(u)) < h_1(\chi(w))$ for all $w \sim_{G_0} u$, in which case $u \in \IS_1$, and hence $v \in N(\IS_1)$. Similarly, by induction, to determine whether $v \in \IS_j$ or $v \in N(\IS_j)$ for some $1 \le j \le \ell$, one only has to know the $(2j)$-th hop neighborhood of $v$ in $G_0$. Therefore, if node $v$ has its $(2\ell)$-th hop neighborhood in $G_0$ stored on the machine associated with $v$, $v$ can determine on its own machine whether $v \in \IS_j$ or $v \in N(\IS_j)$ for some $1 \le j \le \ell$. Hence, $v$ can also determine if it is in $G_{\ell}$.

Therefore, for every node $v$, we run on the machine storing $v$ the algorithm above: for every sequence $\mathfrak{h}$ of hash functions from $\mathcal{H}^*$, check whether $v \in \IS^{\langle \mathfrak{h} \rangle}$ and compute its degree in~$G^{\langle \mathfrak{h} \rangle}$.

As the result, for any sequence $\mathfrak{h} = \langle h_1, \dots, h_{\ell} \rangle$ of hash functions from $\mathcal{H}^*$, every node $v$ in $G_0$ knows its degree $d_{G^{\langle \mathfrak{h} \rangle}}(v)$ in $G^{\langle \mathfrak{h} \rangle}$ and whether $v \in \IS^{\langle \mathfrak{h} \rangle}$. Therefore, using the aggregation approach from Lemma \ref{lem:comm}, in a constant number of \MPC rounds one can find a sequence $\mathfrak{h}_0$ of hash functions from $\mathcal{H}^*$ that minimizes that number of edges in $G^{\langle \mathfrak{h}_0 \rangle}$. This sequence of hash functions is used to define the independent set $\IS^{\langle \mathfrak{h}_0 \rangle}$ found in a given stage. Notice that the resulting graph $G^{\langle \mathfrak{h}_0 \rangle}$ has at most $c_{phase}^\ell \cdot |E_0|$ edges.

\paragraph{Maintaining the $r$-th hop neighborhood at the beginning of stage.}
In our algorithm, we require that at the beginning of each stage, each node $u$ knows its $r$-th hop neighborhood in the current graph, and stores it on the machine associated with $u$. We have discussed how to use a preprocessing, to obtain it at the beginning of the algorithm. Now, suppose that a node $u$ knows its $r$-th hop neighborhood at the beginning of a stage in graph $G_0$; we will show how to ensure that $u$ knows its $r$-th hop neighborhood at the end of that stage, in graph $G_{\ell}$. Notice that if $\IS^*$ is an independent set found in that stage, then $G_{\ell}$ is obtained from $G_0$ by removing $\IS^* \cup N(\IS^*)$. Therefore, to find the $r$-th hop neighborhood of $u$ in $G_{\ell}$, it suffices that $u$ knows all nodes in its $r$-th hop neighborhood in $G_0$ that are in $\IS^* \cup N(\IS^*)$. But this can be easily ensured by asking every node $v \in \IS^* \cup N(\IS^*)$ to send this information to all nodes in its $r$-th hop neighborhood in $G_0$; this can be done in a single \MPC round.

With this, we can summarize our discussion in this section.

\begin{lemma}
Let $\Delta \le n^{\delta}$. Consider an \MPC with $\machines = O(n)$ machines, $\spac = O(n^{\eps})$ space per machine. After a deterministic preprocessing taking $O(\log\log n)$ \MPC rounds, one can run the algorithm to find an MIS in a sequence of stages, such that (i) each stage can be implemented deterministically in a constant number of \MPC rounds, and (ii) each stage reduces the number of edges by at least a factor of $c_{phase}^\ell$, where $\ell = O(\delta \log_{\Delta}n)$.
\qed
\end{lemma}

This immediately bounds the number of stages by $O(\log\Delta)$, yielding Theorem~\ref{thm:MM-MIS-Delta} for MIS.
\qed

\paragraph{Extension to maximal matching.}
As mentioned earlier, the same approach can be used to design a deterministic maximal matching \MPC algorithm in $O(\log\Delta + \log\log n)$ rounds, using $O(n^{\eps})$ space per machine and $O(m+n^{1+\eps})$ total space; since we are able to gather neighborhoods of super-constant radius onto machines, we can perform maximal matching by simulating MIS on the line graph.
%One has to consider only the case when $\Delta \le n^{\eps/200}$. In the same setting as above, the deterministic algorithm relies on a randomized algorithm, in which an edge $e$ is added to the matching iff $h_i(\chi(e)) < h_i(\chi(e'))$ for all $e' \sim e$.
This completes the proof of Theorem~\ref{thm:MM-MIS-Delta}.
\qed

\junk{%%% ???
\paragraph{Maximal matching.}

%noting, as for Corollary \ref{cor:congcMM}, that for maximal matching one can use the standard reduction of performing MIS on the line graph, when $\log \Delta = o(\log n)$
For maximal matching one can use the standard reduction of performing MIS on the line graph, when $\log \Delta = o(\log n)$.
Here, one must assign machines to determine whether nodes of the input graph are removed rather than
edges, in order to obtain $O(m + n^{1+\epsilon})$ rather than $O(m n^{\epsilon})$ total space. That is, machines representing nodes in the original graph can collect balls large enough to simulate the MIS algorithm on all of the adjacent edges, and so we must only collect balls per node in $G$ rather than per edge.
%
%For (deg+1)-list coloring, we use the reduction by Luby \cite{Luby86}. Although the original reduction concerns with the $(\Delta+1)$ coloring, it can easily adapted to the (deg+1) list coloring problem. For completeness we briefly mention it here. The reduction is based on computing a new graph $G'$ with $O(\sum_{u \in V} \deg(u))$ nodes and $O(\sum_{u}\deg^2(u))$ edges. The graph $G'$ is defined as follows. For each node $u$, there is a corresponding clique $C_u$ on $\deg(u)+1$ copies in $G'$. The $i$'th node in $C_u$ corresponds to the $i$'th
%color in $u$'s palette. For every edge $(u,v)$ in the graph $G$, we connect the $i$'th node in $C_u$ with the $j$'th node of $C_v$ provided that the $i^{th}$ color in $u$'s palette is the same as the $j$'th color in $v$'s palette. One can check that any MIS for $G'$ must contain exactly one node in each clique $C_u$. The list-coloring then computes an MIS on $G'$, and the final color of each node $u$ corresponds to the unique node in $C_u$ that appears in this output MIS.
}
%---------------------------------------------------------------------------------------------------------------------------------------------------------

%---------------------------------------------------------------------------------------------------------------------------------------------------------

%---------------------------------------------------------------------------------------------------------------------------------------------------------

%---------------------------------------------------------------------------------------------------------------------------------------------------------
%
%\input{improved-mis.tex}

%---------------------------------------------------------------------------------------------------------------------------------------------------------

\section{Conclusions}
In this paper we study the power of deterministic algorithms on the nowadays classical model of parallel computations --- the \emph{Massively Parallel Computations (\MPC) model} --- on the example of two fundamental graph problems: maximal matching and maximal independent set. We develop a new deterministic method for graph sparsification while maintaining some desired properties and apply it to design the first $O(\log \Delta+\log\log n)$-round fully scalable deterministic \MPC algorithms for maximal matching and MIS (Theorem~\ref{thm:MM-MIS-Delta}). In combination with previous results, this also gives the first deterministic $O(\log \Delta)$-round \congc algorithms for maximal matching and MIS. % (Corollary \ref{cor:congcMISMM}).
We expect our method of derandomizing the sampling of a low-degree graph while maintaining good properties will prove useful for derandomizing many more problems in low space or limited bandwidth models (e.g., the \congest model).

% THIS IS GREAT ... BUT WE DON'T HAVE ANY MORE SPACE
%An interesting open question that remains is whether efficient derandomization is possible in \MPC low-space and \emph{optimal} total space (i.e., $O(m+n)$ rather than $O(m+n^{1+\eps})$). We note that our graph sparsification procedures in all of our algorithms can be implemented in optimal total space. The extra $n^{1+\epsilon}$ total space is required for the efficient derandomization procedures applied on the computed sparse subgraphs. Specifically, these procedures are considerably more efficient provided that the $2$-hop neighborhood of each node is stored at a single machine. We note that our algorithms for MIS and matching should yield $O(\log^2 \Delta)$ rounds when restricting to global space of $\widetilde{O}(m)$. It will be very interesting to provide an improved derandomization techniques for this setting as well.

%---------------------------------------------------------------------------------------------------------------------------------------------------------
\newcommand{\Proc}{Proceedings of the~}
\newcommand{\DISC}{International Symposium on Distributed Computing (DISC)}
\newcommand{\FOCS}{IEEE Symposium on Foundations of Computer Science (FOCS)}
\newcommand{\ICALP}{Annual International Colloquium on Automata, Languages and Programming (ICALP)}
\newcommand{\IPCO}{International Integer Programming and Combinatorial Optimization Conference (IPCO)}
\newcommand{\ISAAC}{International Symposium on Algorithms and Computation (ISAAC)}
\newcommand{\JACM}{Journal of the ACM}
\newcommand{\NIPS}{Conference on Neural Information Processing Systems (NeurIPS)}
\newcommand{\OSDI}{Conference on Symposium on Opearting Systems Design \& Implementation (OSDI)}
\newcommand{\PODS}{ACM SIGMOD Symposium on Principles of Database Systems (PODS)}
\newcommand{\PODC}{ACM Symposium on Principles of Distributed Computing (PODC)}
\newcommand{\SICOMP}{SIAM Journal on Computing}
\newcommand{\SIROCCO}{International Colloquium on Structural Information and Communication Complexity}
\newcommand{\SODA}{Annual ACM-SIAM Symposium on Discrete Algorithms (SODA)}
\newcommand{\SPAA}{Annual ACM Symposium on Parallel Algorithms and Architectures (SPAA)}
\newcommand{\STACS}{Annual Symposium on Theoretical Aspects of Computer Science (STACS)}
\newcommand{\STOC}{Annual ACM Symposium on Theory of Computing (STOC)}
%---------------------------------------------------------------------------------------------------------------------------------------------------------

\bibliographystyle{plain}
\bibliography{derand}

%---------------------------------------------------------------------------------------------------------------------------------------------------------

%---------------------------------------------------------------------------------------------------------------------------------------------------------
\junk{
\clearpage
\onecolumn
\appendix
\begin{center}\huge\bf Appendix\end{center}

Because of space limitations, some explanations and some parts of our analysis are deferred to Appendix.

\section{\mbox{Comparison with earlier work on derandomization by Censor-Hillel et al.\ \cite{CPS17}}}
\label{subsubsec:comp-Censor-Hillel-et-al}
\APPENDCOMPARISONWITHCENSORHILLELETAL

\section{Invariants satisfied by $\E_j$: Proofs of Lemmas \ref{lemma-inv-i}--\ref{lemma-inv-ii} from Section \ref{subsubsec:Properties-of-Ej}}
%Section \ref{sec:max-matching}}
\label{app:matching}
\APPENDMATCHINGINVONE
\APPENDMATCHINGINVTWO

\section{Proof of Lemma \ref{lem:egoodsample} from Section \ref{subsec:CE->M}}
\label{app:egoodsample}

As claimed in Section \ref{subsec:CE->M}, we prove now Lemma \ref{lem:egoodsample}, a useful property analogous to Lemma \ref{lem:Luby} in $\CE$.

\APPENDMATCHINGGOODSAM

\section{Proof of Lemma \ref{lem:egoodhash} from Section \ref{subsec:CE->M}}
\label{app:TwoHop}

\APPENDMATCHINGTWOHOP

\section{Invariants satisfied by $\J_j$: Proofs of Lemmas \ref{lemma-mis-inv-i}--\ref{lemma-mis-inv-ii} from Section \ref{subsub:properties-of-Jj-satisfying-invariant}}
\label{app:mis}
\APPENDMISINVONE
\APPENDMISINVTWO

\section{Proof of Lemma \ref{lem:goodsample} from Section \ref{subsec:->IS}}
\label{app:goodsample}

\APPENDANOTHERONE

\section{Proof of Lemma \ref{lem:goodsample2} from Section \ref{subsec:->IS}}
\label{app:goodsample2}

\APPENDMISSAM

\section{Proof of Lemma \ref{lem:goodhash} from Section \ref{subsec:->IS}}
\label{app:goodhash}

\APPENDMISHASH

%
%\section{Missing Proofs for Section \ref{sec:coloring}}\label{app:coloring}
%\input{color-app.tex}
}
%---------------------------------------------------------------------------------------------------------------------------------------------------------

\end{document}